\documentclass[11pt]{article}

\usepackage{amssymb}
\usepackage{ textcomp }
\usepackage[margin=1in]{geometry}
\usepackage{complexity}
\usepackage{mathtools}
\usepackage{xspace}
\usepackage{rpmacros}
\RequirePackage[colorlinks=true]{hyperref}
\hypersetup{
  linkcolor=[rgb]{0,0,0.4},
  citecolor=[rgb]{0, 0.4, 0},
  urlcolor=[rgb]{0.6, 0, 0}
}
\usepackage{mathpazo}
\usepackage{bbm}
\usepackage{dsfont}

\usepackage{amsthm}
\usepackage{thmtools,thm-restate}

\numberwithin{equation}{section}
\declaretheoremstyle[bodyfont=\it,qed=\qedsymbol]{noproofstyle}

\declaretheorem[numberlike=equation]{observation}

\declaretheorem[name=Observation,numbered=no]{observation*}

\declaretheorem[numberlike=equation]{fact}

\declaretheorem[numberlike=equation]{problem}

\declaretheorem[numberlike=equation]{theorem}

\declaretheorem[name=Theorem,numbered=no]{theorem*}

\declaretheorem[numberlike=equation]{lemma}
\declaretheorem[name=Lemma,numbered=no]{lemma*}

\declaretheorem[numberlike=equation]{corollary}
\declaretheorem[name=Corollary,numbered=no]{corollary*}

\declaretheorem[name=Proposition,numbered=no]{proposition*}

\declaretheorem[numberlike=equation]{claim}
\declaretheorem[name=Claim,numbered=no]{claim*}

\declaretheorem[numberlike=equation]{conjecture}
\declaretheorem[name=Conjecture,numbered=no]{conjecture*}

\declaretheorem[name=Question,numbered=no]{question*}

\declaretheoremstyle[bodyfont=\it,qed=$\lozenge$]{defstyle} 

\declaretheorem[numberlike=equation,style=defstyle]{definition}
\declaretheorem[unnumbered,name=Definition,style=defstyle]{definition*}

\declaretheorem[unnumbered,name=Example,style=defstyle]{example*}

\declaretheorem[unnumbered,name=Notation=defstyle]{notation*}

\declaretheorem[unnumbered,name=Construction,style=defstyle]{construction*}

\declaretheorem[numberlike=equation,style=defstyle]{remark}
\declaretheorem[unnumbered,name=Remark,style=defstyle]{remark*}

\usepackage{nth}
\usepackage{intcalc}
\usepackage{etoolbox}
\usepackage{xstring}
\hypersetup{
}

\usepackage{ifpdf}
\ifpdf
\else
\usepackage[quadpoints=false]{hypdvips}
\fi

\newcommand{\ECCCurl}[2]{http://eccc.hpi-web.de/report/\ifnumcomp{#1}{>}{93}{19}{20}#1/#2/}

\newcommand{\shortECCC}[2]{\texttt{\href{http://eccc.hpi-web.de/report/\ifnumcomp{#1}{>}{93}{19}{20}#1/#2/}{eccc:TR#1-#2}}}

\newcommand{\parseECCC}[1]{
\StrSubstitute{#1}{TR}{}[\tmpstring]%
\IfSubStr{\tmpstring}{/}{ 
\StrBefore{\tmpstring}{/}[\ecccyear]%
\StrBehind{\tmpstring}{/}[\ecccreport]%
}{
\StrBefore{\tmpstring}{-}[\ecccyear]%
\StrBehind{\tmpstring}{-}[\ecccreport]%
}%
\shortECCC{\ecccyear}{\ecccreport}}

\usepackage{algorithmicx}
\usepackage{algorithm} 
\usepackage[noend]{algpseudocode}

\algrenewcommand\algorithmicindent{1.0em}%
\usepackage[T1]{fontenc}
\usepackage[inline]{enumitem}

\newcommand{\vecalpha}{\boldsymbol{\alpha}}
\newcommand{\vecbeta}{\boldsymbol{\beta}}
\newcommand{\vecgamma}{\boldsymbol{\gamma}}

\newcommand{\eqdef}{\vcentcolon=}

\newcommand{\res}{\mathrm{Res}}

	\renewcommand{\vec}[1]{{\mathbf{#1}}}

	\makeatletter
	\newcommand{\va}{{\vec{a}}\@ifnextchar{^}{\!\:}{}}
	\newcommand{\vb}{{\vec{b}}\@ifnextchar{^}{\!\:}{}}
	\newcommand{\vc}{{\vec{c}}\@ifnextchar{^}{\!\:}{}}
	\newcommand{\vd}{{\vec{d}}\@ifnextchar{^}{\!\:}{}}
	\newcommand{\ve}{{\vec{e}}\@ifnextchar{^}{\!\:}{}}
	\newcommand{\vy}{{\vec{y}}\@ifnextchar{^}{\!\:}{}}
	\newcommand{\vs}{{\vec{s}}\@ifnextchar{^}{\!\:}{}}
	\newcommand{\vt}{{\vec{t}}\@ifnextchar{^}{\!\:}{}}
	\newcommand{\vx}{{\vec{x}}\@ifnextchar{^}{}{}}		
	\newcommand{\vz}{{\vec{z}}\@ifnextchar{^}{\!\:}{}}
	\newcommand{\vv}{{\vec{v}}\@ifnextchar{^}{\!\:}{}}
	\newcommand{\vu}{{\vec{u}}\@ifnextchar{^}{\!\:}{}}
	\newcommand{\vf}{{\vec{f}}\@ifnextchar{^}{\!\:}{}}
	\newcommand{\vg}{{\vec{g}}\@ifnextchar{^}{\!\:}{}}
	\newcommand{\vr}{{\vec{r}}\@ifnextchar{^}{\!\:}{}}
	\newcommand{\vw}{{\vec{w}}\@ifnextchar{^}{\!\:}{}}

	\newcommand{\vY}{{\vec{Y}}\@ifnextchar{^}{\!\:}{}}
	\newcommand{\vX}{{\vec{X}}\@ifnextchar{^}{}{}}		
	\newcommand{\vZ}{{\vec{Z}}\@ifnextchar{^}{\!\:}{}}
	\newcommand{\vG}{{\vec{G}}\@ifnextchar{^}{\!\:}{}}
	
	\makeatother

	\newcommand{\vaa}{{\vecalpha}}
	\newcommand{\vbb}{{\vecbeta}}
	\newcommand{\vcc}{{\vecgamma}}
	
	\newcommand{\vxx}{\boldsymbol{\xi}}
	
\renewcommand{\C}{\mathbb{C}}
\renewcommand{\N}{\mathbb{N}}

\newcommand{\cA}{{\mathcal{A}}}
\newcommand{\cB}{{\mathcal{B}}}
\newcommand{\cC}{{\mathcal{C}}}

\newcommand{\cF}{{\mathcal{F}}}
\newcommand{\cI}{{\mathcal{I}}}
\newcommand{\cJ}{{\mathcal{J}}}
\newcommand{\cK}{{\mathcal{K}}}
\newcommand{\calL}{{\mathcal{L}}}
\newcommand{\cQ}{{\mathcal{Q}}}
\newcommand{\calP}{{\mathcal{P}}}

\newcommand{\cT}{{\mathcal{T}}}

\newcommand{\calS}{{\mathcal{S}}}

\newcommand{\ideal}[1]{\left \langle{#1}\right \rangle}
\newcommand{\MVar}[2]{{#1}_1,\ldots, {#1}_{#2}}

\newcommand{\PRing}[3]{\mathbb{#1}[ \MVar{#2}{#3}]}

\newcommand{\CRing}[2]{\PRing{C}{#1}{#2}}

\newcommand{\spn}[1]{\operatorname{span}\{{#1}\}}

\def\epsilon{\varepsilon} 

\newcommand{\shir}[1]{\newline{\color{blue} Shir Replay: {#1}}}

\newcommand{\rkq}{100}
\newcommand{\MS}{\text{Lin}}

\newif\ifEK

\date{}

\title{A generalized Sylvester-Gallai type theorem for quadratic polynomials}
\author{Shir Peleg\thanks{Department of Computer Science, Tel Aviv University, Tel Aviv, Israel, E-mail: \texttt{shirpele@mail.tau.ac.il, shpilka@tauex.tau.ac.il}. The research leading to these results has received funding from the  Israel Science Foundation (grant number 552/16) and from the Len Blavatnik and the Blavatnik Family foundation. Part of this work was done while the second author was a visiting professor at NYU.
}   \and  Amir Shpilka\footnotemark[1]}
\begin{document}
\maketitle
\begin{abstract}

In this work we prove a version of the Sylvester-Gallai theorem for quadratic polynomials that takes us one step closer to obtaining a deterministic polynomial time algorithm for testing zeroness of $\Sigma^{[3]}\Pi\Sigma\Pi^{[2]}$ circuits. Specifically, we prove that if a finite set of irreducible quadratic polynomials $\cQ$ satisfy that for every two polynomials $Q_1,Q_2\in \cQ$ there is a subset  $\cK\subset \cQ$, such that $Q_1,Q_2 \notin \cK$ and whenever $Q_1$ and $Q_2$ vanish then also $\prod_{i\in \cK} Q_i$ vanishes, then the linear span of the polynomials in  $\cQ$ has dimension $O(1)$. This extends the earlier result~\cite{DBLP:conf/stoc/Shpilka19}  that showed a similar conclusion when $|\cK| = 1$.

An important technical step in our proof is a theorem classifying all the possible cases in which  a product of quadratic polynomials can vanish when two other quadratic polynomials vanish. I.e., when the product is in the radical of the ideal generates by the two quadratics. This step extends a result from~\cite{DBLP:conf/stoc/Shpilka19} that studied the case when one quadratic polynomial is in the radical of two other quadratics.
\end{abstract}

\thispagestyle{empty}
\newpage

\tableofcontents

\thispagestyle{empty}
\newpage
\pagenumbering{arabic}

\section{Introduction}\label{sec:intro}

This paper studies a problem at the intersection of algebraic complexity, algebraic geometry and combinatorics that is motivated by the polynomial identity testing problem (PIT for short) for depth $4$ circuits. The question can also be regarded as an algebraic generalization and extension of the famous Sylvester-Gallai theorem from discrete geometry. We shall first describe the Sylvester-Gallai theorem and some of its many extensions and generalization and then discuss the relation to PIT.

\paragraph{Sylvester-Gallai type theorems:}
The Sylvester-Gallai theorem asserts that if a finite set of points in $\R^n$ has the property that every line passing through any two points in the set also contains a third point in the set then all the points in the set are colinear.
Kelly extended the theorem to points in $\C^n$ and proved that if a finite set of points satisfy the Sylvester-Gallai condition then  the points in the set are coplanar. Many variants of this theorem were studied: extensions to higher dimensions, colored versions, robust versions and many more. For a more on the Sylvester-Gallai theorem and some of its variants see \cite{BorwenMoser90,barak2013fractional,DSW12}. 

There are two extensions that are of specific interest for our work:
The \textbf{colored} version, proved by Edelstein and Kelly, states that if three finite sets of points satisfy that every line passing through points from two different sets also contains a point from the third set, then, all the points belong to a low dimensional space. This result was further extended to any constant number of sets.
The \textbf{robust} version, obtained in \cite{barak2013fractional,DSW12}, states that if a finite set of points satisfy that for every point $p$ in the set a $\delta$ fraction of the other points satisfy that the line passing through each of them and $p$ spans a third point in the set, then the set is contained in an $O(1/\delta)$-dimensional space. 

Although the Sylvester-Gallai theorem is formulated as a geometric question, it can be stated in  algebraic terms:
If a finite set of pairwise linearly independent vectors, $\calS \subset \C^n$, has the property that every two vectors span a third vector in the set then the dimension of $\calS$ is at most $3$. It is not very hard to see that if we pick a subspace $H$ of codimension $1$, which is in general position with respect to the vectors in the set, then the intersection points $p_i = H \cap \spn{s_i}$, for $s_i\in \calS$, satisfy the Sylvester-Gallai condition. Therefore, $\dim(S) \leq 3$.  
Another formulation is the following: If a finite set of pairwise linearly independent linear forms, $\calL \subset \CRing{x}{n}$, has the property that for every two forms $\ell_i,\ell_j \in \calL$ there is a third form $\ell_k \in \calL$, so that whenever $\ell_i$ and $\ell_j$ vanish then so does $\ell_k$,  then the linear dimension
of $\calL$ is at most $3$. To see this note that it must be the case that $\ell_k \in \spn{\ell_i,\ell_j}$ and thus the coefficient vectors of the forms in the set satisfy the condition for the (vector version of the) Sylvester-Gallai theorem, and the bound on the dimension follows.


The last formulation can now be extended to higher degree polynomials. In particular, the following question was asked by Gupta \cite{Gupta14}.
 
\begin{problem}\label{prob:sg-alg}
Can we bound the linear dimension or algebraic rank of a finite set $\calP$ of pairwise linearly independent irreducible polynomials of degree at most $r$ in $\CRing{x}{n}$, that has the following property: For any two distinct polynomials $P_1,P_2 \in \calP$ there is a third polynomial $P_3 \in \calP$, such that whenever $P_1,P_2$ vanish then so does $P_3$. 
\end{problem}

A robust or colored version of this problem can also be formulated.
As we have seen, the case  $r=1$, i.e when all the polynomials are linear forms, follows from the Sylvester-Gallai theorem. For the case of quadratic polynomials, i.e. $r=2$,  \cite{DBLP:conf/stoc/Shpilka19} gave a bound on the linear dimension for both the non-colored and colored versions. A bound for the robust version is still unknown for $r=2$ and the entire problem is open for $r \geq 3$.
Gupta \cite{Gupta14} also raised a more general question of the same form.
 
\begin{problem}\label{prob:sg-alg-prod}
Can we bound the linear dimension or algebraic rank of a finite set $\calP$ of pairwise linearly independent irreducible polynomials of degree at most $r$ in $\CRing{x}{n}$  that has the following property: For any two distinct polynomials $P_1,P_2 \in \calP$ there is a subset $\cI \subset \calP$, such that $P_1,P_2 \notin \cI$ and whenever $P_1,P_2$ vanish then so does $\prod_{P_i\in \cI} P_i$. 
\end{problem}

As before this problem can also be extended to robust and colored versions.
In the case of linear forms, the bound for \autoref{prob:sg-alg} carries over to \autoref{prob:sg-alg-prod} as well. This follows from the fact that the ideal generated by linear forms is prime (see \autoref{sec:prelim} for definitions). In the case of higher degree polynomials, there is no clear reduction. For example, let $r=2$ and 
$$P_1 = xy+zw \quad,\quad P_2 = xy-zw \quad,\quad P_3 = xw \quad,\quad P_4 = yz.$$
Then, it is not hard to verify that whenever $P_1$ and $P_2$ vanish then so does $P_3 \cdot P_4$, 
but neither $P_3$ nor $P_4$ always vanishes when $P_1$ and $P_2$ do. The reason is that the radical of the ideal generated by $P_1$ and $P_2$ is not prime. Thus it is not clear whether a bound for \autoref{prob:sg-alg} would imply a bound for \autoref{prob:sg-alg-prod}. The latter problem was
open, prior to this work, for any degree $r > 1$.


The Sylvester-Gallai theorem has important consequences for locally decodable and locally correctable codes \cite{barak2013fractional,DSW12}, for reconstruction of certain depth-$3$ circuits \cite{DBLP:journals/siamcomp/Shpilka09,DBLP:conf/coco/KarninS09,DBLP:conf/coco/Sinha16} and for the polynomial identity testing (PIT for short) problem, which we describe next.

\paragraph{Sylvester-Gallai type theorems and PIT:}

The PIT problem asks to give a deterministic algorithm that given an arithmetic circuit as input determines whether it computes the identically zero polynomial. This is a fundamental problem in theoretical computer science that has attracted a lot of attention  because of its intrinsic importance, its relation to other derandomization problems \cite{DBLP:journals/cc/KoppartySS15,Mulmuley-GCT-V,DBLP:conf/approx/ForbesS13,DBLP:journals/cacm/FennerGT19,DBLP:conf/stoc/GurjarT17,DBLP:conf/focs/SvenssonT17} and its connections to lower bounds for arithmetic circuits \cite{DBLP:conf/stoc/HeintzS80,DBLP:conf/fsttcs/Agrawal05,DBLP:journals/cc/KabanetsI04,DBLP:journals/siamcomp/DvirSY09,DBLP:journals/toc/ForbesSV18,DBLP:conf/coco/ChouKS18}. 
Perhaps surprisingly, it was shown that deterministic algorithms for the PIT problem for homogeneous depth-$4$ circuits or for depth-$3$ circuits would lead to deterministic algorithms for general circuits \cite{DBLP:conf/focs/AgrawalV08,DBLP:conf/focs/0001KKS13}. This makes small depth circuit extremely interesting for the PIT problem. We next explain how Sylcester-Gallai type questions are directly related to PIT for such low depth circuits.
For more on the PIT problem see  \cite{DBLP:journals/fttcs/ShpilkaY10,Saxena09,Saxena14, ForbesThesis}. 

The Sylvester-Gallai theorem is mostly relevant for the PIT problem in the  setting when the input is a depth-$3$ circuit with small top fan-in. Specifically, a homogeneous $\Sigma^{[k]}\Pi^{[d]}\Sigma$ circuit in $n$ variables computes a polynomial of the  form
\begin{equation}\label{eq:sps}
\Phi(x_1,\ldots,x_n) = \sum_{i=1}^{k}\prod_{j=1}^{d}\ell_{i,j}(x_1,\ldots,x_n)\;,
\end{equation}
where each $\ell_{i,j}$ is a linear form.
Consider the PIT problem for $\Sigma^{[3]}\Pi^{[d]}\Sigma$ circuits, i.e., $\Phi$ is given as in \autoref{eq:sps} and  $k=3$. In particular, 
\begin{equation}\label{eq:sps3}
\Phi(x_1,\ldots,x_n) = \prod_{j=1}^{d}\ell_{1,j}(x_1,\ldots,x_n)+\prod_{j=1}^{d}\ell_{2,j}(x_1,\ldots,x_n)+\prod_{j=1}^{d}\ell_{3,j}(x_1,\ldots,x_n)\;.
\end{equation}
If $\Phi$ computes the zero polynomial, then for every $j,j'\in[d]$.
$$\prod_{i=1}^{d}\ell_{1,i} \equiv 0 \mod \ideal{\ell_{2,j},\ell_{3,j'}}\;.\footnote{By $\ideal{\ell_{2,j},\ell_{3,j'}}$ we mean the ideal generated by $\ell_{2,j}$ and $\ell_{3,j'}$. See \autoref{sec:prelim}.}$$
This means that the sets $\cT_i = \{\ell_{i,1}, \ldots, \ell_{i,d}\}$ satisfy the conditions of the colored version of \autoref{prob:sg-alg-prod} for $r=1$, and therefore have a small linear dimension. Thus, if $\Phi\equiv 0$ then, assuming that no linear form belongs to all three sets, we can rewrite the expression for $\Phi$ using only constantly many variables (after a suitable invertible linear transformation). This gives an efficient PIT algorithms for such $\Sigma^{[3]}\Pi^{[d]}\Sigma$ identities. The case of more than three multiplication gates is more complicated but it also satisfies a similar higher dimensional condition. This rank-bound approach for PIT of $\Sigma\Pi\Sigma$ circuits was raised in \cite{DBLP:journals/siamcomp/DvirS07} and later carried out in \cite{DBLP:conf/focs/KayalS09,DBLP:journals/jacm/SaxenaS13}.\footnote{The best algorithm for PIT of $\Sigma^{[k]}\Pi^{[d]}\Sigma$  circuits was obtained through a different, yet related, approach in \cite{DBLP:journals/siamcomp/SaxenaS12}.}

\sloppy As such rank-bounds found important applications in studying PIT of depth-$3$ circuits it seemed that a similar approach could potentially work for depth-$4$ $\Sigma\Pi\Sigma\Pi$ circuits as well.\footnote{For multilinear $\Sigma\Pi\Sigma\Pi$ circuits Saraf and Volkovich obtained an analogous bound on the sparsity of the polynomials computed by the multiplication gates in a zero circuit \cite{DBLP:journals/combinatorica/SarafV18}.} In particular, it seemed most relevant for the case where there are only three multiplication gates and the bottom fan-in is two, i.e. for homogeneous  $\Sigma^{[3]}\Pi^{[d]}\Sigma\Pi^{[2]}$ circuits that compute polynomials of the form
\begin{equation}\label{eq:spsp}
\Phi(x_1,\ldots,x_n) = \prod_{j=1}^{d}Q_{1,j}(x_1,\ldots,x_n)+\prod_{j=1}^{d}Q_{2,j}(x_1,\ldots,x_n)+\prod_{j=1}^{d}Q_{3,j}(x_1,\ldots,x_n)\;.
\end{equation}
Both Beecken et al.  \cite{DBLP:journals/iandc/BeeckenMS13} and Gupta \cite{Gupta14} suggested an approach to the PIT problem of such identities  based on the colored version of \autoref{prob:sg-alg-prod} for $r=2$. Both papers described PIT algorithms for depth-$4$ circuits assuming a bound on the algebraic rank of the polynomials. 
In fact, Gupta conjectured that the algebraic rank of polynomials satisfying the conditions of \autoref{prob:sg-alg-prod} depends only on their degree (see Conjectures $1,2$ and $30$ in \cite{Gupta14}). 

\begin{conjecture}[Conjecture 1 in \cite{Gupta14}]\label{con:gupta-general}
Let $\cF_1,\ldots, \cF_k$ be finite sets of irreducible homogenous polynomials in $\C[x_1,\ldots, x_n]$ of degree $\leq r$ such that $\cap_i \cF_i = \emptyset$ and for every $k-1$ polynomials $Q_1,\ldots,Q_{k-1}$, each from a distinct set, there are $P_1,\ldots,P_c$ in the remaining set such that  whenever $Q_1,\ldots,Q_{k-1}$ vanish then also the product $\prod_{i=1}^{c}P_i$ vanishes. Then, $\text{trdeg}_\C(\cup_i \cF_i) \leq \lambda(k, r, c)$ for some function $\lambda$, where trdeg stands for the transcendental degree (which is the same as algebraic rank).
\end{conjecture}

Furthermore, using degree arguments Gupta showed  that in \autoref{prob:sg-alg-prod} we can restrict our attention to sets $\cI$ such that $|\cI|\leq r^{k-1}$.
In particular, if the circuit in Equation~\eqref{eq:spsp} vanishes identically, then for every $(j,j')\in[d]^2$ there are $i_{1,j,j'}, i_{2,j,j'},i_{3,j,j'}, i_{4,j,j'}\in [d]$ so that 
$$Q_{1,i_{1,j,j'}}\cdot Q_{1,i_{2,j,j'}} \cdot Q_{1,i_{3,j,j'}}\cdot Q_{1,i_{4,j,j'}} \equiv 0 \mod \ideal{Q_{2,j},Q_{3,j'}}.$$

\sloppy In \cite{DBLP:journals/iandc/BeeckenMS13} Beecken et al. conjectured that  the algebraic rank of simple and minimal $\Sigma^{[k]}\Pi^{[d]}\Sigma\Pi^{[r]}$ circuits (see their paper for definition of simple and minimal) is $O_k(\log d)$. We note that for $k=3$ this conjecture is weaker than  Conjecture~\ref{con:gupta-general} as every zero $\Sigma^{[3]}\Pi^{[d]}\Sigma\Pi^{[r]}$ circuit gives rise to a structure satisfying the conditions of \autoref{con:gupta-general}, but the other direction is not necessarily true.
Beecken et al. also showed how to obtain a deterministic PIT for $\Sigma^{[k]}\Pi^{[d]}\Sigma\Pi^{[r]}$ circuits, assuming the correctness of their conjecture.
 


\subsection{Our Result}
%

Our main result gives a bound on the linear dimension of polynomials satisfying the conditions of \autoref{prob:sg-alg-prod} when all the polynomials are irreducible of degree at most $2$. Specifically we prove the following theorem. 

\begin{theorem}\label{thm:main-sg-intro}
There exists a universal constant $c$ such that the following holds. 
Let $\tilde{\cQ} = \{Q_i\}_{i\in \{1,\ldots,m\}}\subset\C[x_1,\ldots,x_n]$ be a finite set of pairwise linearly independent irreducible polynomials of degree at most $2$. Assume  that, for every $i\neq j$, whenever $Q_i$ and $Q_j$ vanish then so does $\prod_{k\in  \{1,\ldots,m\} \setminus\{i,j\}} Q_k$. Then,  $\dim(\spn{\cQ})\leq c$.
\end{theorem}

While our result still does not resolve \autoref{con:gupta-general}, as we need a colorful version of it, 
we believe that it is a significant step towards solving the conjecture for $k=3$ and $r=2$, which will yield  a PIT algorithm for $\Sigma^{[3]}\Pi^{[d]}\Sigma\Pi^{[2]}$ circuits.  

An interesting aspect of our result is that while the conjectures of  \cite{DBLP:journals/iandc/BeeckenMS13,Gupta14} speak about the algebraic rank we prove a stronger result that bounds that linear dimension (the linear rank is an upper bound on the algebraic rank). As our proof is quite technical it is an interesting question whether one could simplify our arguments by arguing directly about the algebraic rank.

An important algebraic tool in the proof of  \autoref{thm:main-sg-intro} is the following result  characterizing the different cases in which a product of quadratic polynomials vanishes whenever two other quadratics vanish. 

\begin{theorem}\label{thm:structure-intro}
Let $\{Q_k\}_{k\in \cK},A$ and $B$ be $n$-variate, homogeneous, quadratic polynomials, over $\C$, satisfying that whenever $A$ and $B$ vanish then so does $\prod_{k\in \cK} Q_k$. Then, one of the following cases must hold:
\begin{enumerate}[label={(\roman*)}]
\item There is $k\in \cK$ such that $Q_k$ is in the linear span of $A$ and $B$.  \label{case:span-intro}
\item \label{case:rk1-intro}
There exists a non trivial linear combination of the form $\alpha A+\beta B = ab$ where  $a$ and $b$ are linear forms. 
\item There exist two linear forms $a$ and $b$ such that when setting $a=b=0$ we get that $A$ and $B$ vanish. 
\end{enumerate}
\end{theorem}

The statement of the result is quite similar to Theorem 1.8 of  \cite{DBLP:conf/stoc/Shpilka19} that proved a similar result when $|\cK|=1$. Specifically, in \cite{DBLP:conf/stoc/Shpilka19} the second item reads ``There exists a non trivial linear combination of the form $\alpha A+\beta B = a^2$, where $a$ is a linear form.'' This ``minor'' difference in the statements (which is necessary) is also responsible for the much harder work we do in the paper. 


%

\subsection{Proof Idea}\label{sec:proof-idea}

Our proof has a similar structure to the proofs in \cite{DBLP:conf/stoc/Shpilka19}, but it does not rely on any of the results proved there. 

Our starting point is the observation that \autoref{thm:structure-intro} guarantees that unless one of $\{Q_k\}$ is in the linear span of $A$ and $B$ then $A$ and $B$ must satisfy a very strong property, namely, they must span a reducible quadratic or they have a very low rank (as quadratic polynomials). The proof of this theorem is based on analyzing the resultant of $A$ and $B$ with respect to some variable. We now explain how this theorem can be used to prove \autoref{thm:main-sg-intro}.

Consider a set of polynomials $\cQ=\{\MVar{Q}{m}\}$ satisfying the condition of \autoref{thm:main-sg-intro}. First, consider the case in which for every $Q\in\cQ$, at least, say, $(1/100)\cdot m$ of the polynomials  $Q_i\in \cQ$, satisfy that there is another polynomial in $\cQ$ in $\spn{Q,Q_i}$. In this case, we can use the robust version of the Sylvester-Gallai theorem \cite{barak2013fractional, DSW12} (see \autoref{thm:robustSG}) to deduce that the linear dimension of $\cQ$  is small.

The second case we consider is when every polynomial $Q\in\cQ$ that did not satisfy the first case now  satisfies that for at least, say, $(1/100)\cdot m$ of the polynomials  $Q_i\in \cQ$ there are linear forms $a_i$ and $b_i$ such that $Q,Q_i \in \ideal{a_i,b_i}$. We prove that if this is the case then there is a bounded dimensional linear space of linear forms, $V$, such that all the polynomials in $\cQ$ that are of rank $2$ are in $\ideal{V}$. Then we argue that the polynomials that are not in $\ideal{V}$ satisfy the robust version of the Sylvester-Gallai theorem (\autoref{thm:robustSG}). Finally we bound the dimension of $\cQ \cap \ideal{V}$.

Most of the work however (\autoref{sec:q-dom}) goes into studying what happens in the remaining case when there is some polynomial $Q_o\in\cQ$ for which at least $0.98m$ of the other polynomials in $\cQ$ satisfy \autoref{thm:structure-intro}\ref{case:rk1-intro} with $Q_o$. This puts a  strong restriction on the structure of these $0.98m$ polynomials. Specificity, each of them is of the form $Q_i = Q_o+a_ib_i$, where $a_i$ and $b_i$ are linear forms.
The idea in this case is to show that the set $\{a_i,b_i\}$ is of low dimension. This is done by again studying the consequences of \autoref{thm:structure-intro} for pairs of polynomials $Q_o+a_ib_i,Q_o+a_jb_j\in \cQ$.
After bounding the dimension of these $0.98m$ polynomials we bound the dimension of all the polynomials in $\cQ$. The proof of this case is much more involved than the cases described earlier, and in particular we handle differently the case where $Q_o$ is of high rank and the case where its rank is low.

\subsection{On the relation to the proof of \cite{DBLP:conf/stoc/Shpilka19}}

In \cite{DBLP:conf/stoc/Shpilka19} the following theorem was proved.

\begin{theorem}[Theorem 1.7 of \cite{DBLP:conf/stoc/Shpilka19}]\label{thm:shpilka-main}
Let $\{Q_i\}_{i\in [m]}$ be homogeneous quadratic polynomials over $\C$ such that each $Q_i$ is either irreducible or a square of a linear function. Assume further that for every $i\neq j$ there exists $k\not\in \{i,j\}$ such that whenever $Q_i$ and $Q_j$ vanish $Q_k$ vanishes as well. Then the linear span of the $Q_i$'s has dimension $O(1)$.
\end{theorem}

As mentioned earlier, the steps in our proof are similar to the proof of Theorem 1.7 in  \cite{DBLP:conf/stoc/Shpilka19}. Specifically,  \cite{DBLP:conf/stoc/Shpilka19} also relies on an analog of \autoref{thm:structure-intro} and divides the proof according to whether all polynomials satisfy the first case above or not. However, 
the fact that case~\ref{case:rk1-intro} of \autoref{thm:structure-intro} is different than the corresponding case  in the statement of  Theorem 1.8 of  \cite{DBLP:conf/stoc/Shpilka19}, 
makes our proof is significantly more difficult. The reason for this is that while in \cite{DBLP:conf/stoc/Shpilka19} we could always pinpoint which polynomial vanishes when $Q_i$ and $Q_j$ vanish, here we only know that this polynomial belongs to a small set of polynomials. This leads to a richer structure in \autoref{thm:structure-intro} and consequently to a considerably more complicated proof.
To understand the effect of this on our proof we note that the corresponding case to \autoref{thm:structure-intro}\ref{case:rk1-intro} was the \emph{simpler} case to analyze in  the proof of \cite{DBLP:conf/stoc/Shpilka19}. The fact that $a_i=b_i$ when $|\cK|=1$ almost immediately implied that the dimension of the span of the $a_i$s is constant (see Claim 5.2 in  \cite{DBLP:conf/stoc/Shpilka19}). In our case however, this is the bulk of the proof, and \autoref{sec:q-dom} is devoted to handling  this case. 

In addition to being technically more challenging, our proof gives new insights that may be extended to higher degree polynomials. The first is \autoref{thm:structure-intro}. While a similar theorem was proved for the simpler setting of \cite{DBLP:conf/stoc/Shpilka19}, it was not clear whether a characterization in the form given in \autoref{thm:structure-intro} would be possible, let alone true, in our more general setting. This gives hope that a similar  result would be true for higher degree polynomials. 
Our second contribution is that we show (more or less) that either the polynomials in our set satisfy the robust version of Sylvester-Gallai theorem (\autoref{def:delta-SGConf}) or the linear functions composing the polynomials satisfy the theorem. Potentially, this may be extended to higher degree polynomials. 


\subsection{Organization}
The paper is organized as follows. \autoref{sec:prelim} contains basic facts regarding the resultant and some other  tools and notation used in this work. \autoref{sec:structure} contains the proof of our structure theorem (\autoref{thm:structure-intro}). In \autoref{sec:quad-SG} we give the proof of \autoref{thm:main-sg-intro}. This proof uses a main theorem which will be proved in \autoref{sec:q-dom}. Finally in \autoref{sec:discussion} we discuss further directions and open problems.

\section{Preliminaries}\label{sec:prelim}

In this section we explain our notation and present some basic algebraic preliminaries.

We will use the following notation. Greek letters $\alpha, \beta,\ldots$ denote scalars from $\C$.  Non-capitalized letters $a,b,c,\ldots$  denote linear forms and $x,y,z$ denote variables (which are also linear forms). Bold faced letters denote vectors, e.g. $\vx=(x_1,\ldots,x_n)$ denotes a vector of variables, $\vaa=(\alpha_1,\ldots,\alpha_n)$ is a vector of scalars, and $\vec{0} = (0,\ldots,0)$ the zero vector. We sometimes do not use a boldface notation for a point in a vector space if we do not use its structure as vector. Capital letters such as $A,Q,P$  denote quadratic polynomials whereas $V,U,W$ denote linear spaces. Calligraphic letters $\cal I,J,F,Q,T$  denote sets. For a positive integer $n$ we denote $[n]=\{1,2,\ldots,n\}$. For a matrix $X$ we denote by $|X|$ the determinant of $X$.

A \emph{Commutative Ring} is a group that is abelian with respect to both multiplication and addition operations. We mainly use the multivariate polynomial ring, $\CRing{x}{n}$. An \emph{Ideal} $I\subseteq \CRing{x}{n}$ is an abelian subgroup that is closed under multiplication by ring elements. For $\calS \subset \CRing{x}{n}$, we  denote with $\ideal{\calS}$, the ideal generated by $\calS$, that is, the smallest ideal that contains $\calS$. For example, for two polynomials $Q_1$ and $Q_2$, the ideal $\ideal{Q_1,Q_2}$ is the set $\CRing{x}{n}Q_1 + \CRing{x}{n}Q_2$. For a linear subspace $V$, we have that $\ideal{V}$ is the ideal generated by any basis of $V$. The \emph{radical} of an ideal $I$, denoted by $\sqrt{I}$, is the set of all ring elements, $r$, satisfying that for some natural number $m$ (that may depend on $r$), $r^m \in I$. Hilbert's Nullstellensatz  implies that, in $\C[x_1,\ldots,x_n]$, if a polynomial $Q$ vanishes whenever $Q_1$ and $Q_2$ vanish, then $Q\in \sqrt{\ideal{Q_1,Q_2}}$ (see e.g. \cite{CLO}). We shall often use the notation $Q\in \sqrt{\ideal{Q_1,Q_2}}$ to denote this vanishing condition. For an ideal $I\subseteq \CRing{x}{n}$ we denote by $\CRing{x}{n} /I$ the \emph{quotient ring}, that is, the ring whose elements are the cosets of $I$ in $\CRing{x}{n}$ with the proper multiplication and addition operations. For an ideal $I\subseteq \CRing{x}{n}$ we denote the set of all common zeros of elements of $I$ by $Z(I)$.

For $V_1,\ldots,V_k$ linear spaces, we use $\sum_{i=1}^k V_i$ to denote the linear space $V_1 + \ldots + V_k$. For two non zero polynomials $A$ and $B$ we denote $A\sim B$ if $B \in \spn{A}$. For a space of linear forms $V = \spn{\MVar{v}{\Delta}}$, we say that a polynomial $P \in \CRing{x}{n}$ depends only on $V$  if the value of $P$ is determined by the values of the linear forms $v_1,\ldots,v_\Delta$. More formally, we say that $P$ depends only on $V$ if there is a $\Delta$-variate polynomial $\tilde{P}$ such that $P \equiv \tilde{P}(v_1,\ldots,v_\Delta)$. We denote by  $\CRing{v}{\Delta}\subseteq \CRing{x}{n}$ the subring of polynomials that depend only on $V$.

Another notation that we will use throughout the proof is congruence modulo linear forms.
\begin{definition}\label{def:mod-form}
Let $V\subset \CRing{x}{n}$ be a space of linear forms, and $P,Q\in \CRing{x}{n}$. We say that $P\equiv_V Q$ if $P-Q \in \ideal{V}$. 
\end{definition}

\begin{fact}\label{fact:ufd}
Let $V\subset \CRing{x}{n}$ be a space of linear forms and $P,Q\in \CRing{x}{n}$. If $P = \prod_{k=1}^t P_k$, and $Q = \prod_{k=1}^t Q_k$ satisfy that for all $k$, $P_k$ and $Q_k$ are irreducible in $\CRing{x}{n}/{\langle V\rangle}$, and $P \equiv_V Q \not\equiv_V 0$ then, up to a permutation of the indices, $P_k\equiv_V Q_k$ for all $k\in [t]$.
\end{fact}
This follows from the fact that the quotient ring $\CRing{x}{n}/{\langle V\rangle}$is a unique factorization domain.

\subsection{Sylvester-Gallai Theorem and some of its Variants}\label{sec:SGThms}
In this section we present the formal statement the of Sylvester-Gallai theorem and the extensions that we use in this work.

\begin{definition}
Given a set of points, $v_1,\ldots ,v_m$, we call a line that passes through exactly two of the
points of the set an \emph{ordinary line}.
\end{definition}

\begin{theorem}[Sylvester-Gallai theorem]\label{thm:SG}
If $m$ distinct points $v_1,\ldots ,v_m$ in $\R^n$ are not collinear, then they define at least one ordinary line.
\end{theorem}

\begin{theorem}[Kelly's theorem]\label{thm:Kelly}
If $m$ distinct points $v_1,\ldots ,v_m$ in $\C^n$ are not coplanar, then they define at least one ordinary line.
\end{theorem}

The robust version of the theorem was stated and proved in \cite{barak2013fractional,DSW12}.
\begin{definition}\label{def:delta-SGConf}
We say that a set of points $v_1,\ldots ,v_m\in \C^n$ is a \textit{$\delta$-SG configuration} if for every $i\in [m]$ there exists at least $\delta m$ values of $j \in [m]$ such that the line through $v_i,v_j$ contains a third point in the set.
\end{definition}

\begin{theorem}[Robust Sylvester-Gallai theorem, Theorem $1.9$ of \cite{DSW12}]\label{thm:robustSG}
Let $V = \lbrace v_1,\ldots ,v_m\rbrace \subset \C^n$ be a $\delta$-SG configuration. Then $\dim(\spn{v_1,\ldots ,v_m}) \leq \frac{12}{\delta}+1$.
\end{theorem}

The following is the colored version of the Sylvester-Gallai theorem.

\begin{theorem}[Theorem $3$ of \cite{EdelsteinKelly66}]\label{thm:EK}
Let $\cT_i$, for $i\in [3]$, be disjoint finite subsets of $\C^n$ such that for every $i\neq j$ and any two points $p_1\in \cT_i$ and $p_2 \in \cT_j$ there exists a point $p_3$ in the third set that lies on the line passing through $p_1$ and $p_2$. Then, any such $\cT_i$ satisfy that $\dim(\spn{\cup_i \cT_i})\leq 3$.
\end{theorem}

We also state the equivalent algebraic versions of Sylvester-Gallai.

\begin{theorem}\label{thm:SGVec}
Let $\calS = \lbrace \vec{s}_1,\ldots ,\vec{s}_m\rbrace \subset \C^n$ be a set of pairwise linearly independent vectors such that for every $i\neq j\in [m]$ there is a distinct $k \in [m]$ for which $\vec{s}_k \in \spn{\vec{s}_i,\vec{s}_j}$. Then $\dim(\calS) \leq 3$.
\end{theorem}

\begin{theorem}\label{thm:SG-linforms}
Let $\calP = \lbrace \ell_1,\ldots ,\ell_m\rbrace \subset \CRing{x}{n}$ be a set of pairwise linearly independent linear forms such that for every $i\neq j\in [m]$ there is a distinct $k \in [m]$ for which whenever $\ell_i, \ell_j$ vanish so does $\ell_k$. Then $\dim(\calP) \leq 3$.
\end{theorem}

In this paper we refer to each of \autoref{thm:Kelly}, \autoref{thm:SGVec} and \autoref{thm:SG-linforms} as the Sylvester-Gallai theorem. We shall also refer to sets of points/vectors/linear forms that satisfy the conditions of the relevant theorem as satisfying the condition of the Sylvester-Gallai theorem.

\subsection{Resultant}\label{sec:res}

A tool that will play an important role in the proof of \autoref{thm:structure-intro} is the resultant of two polynomials. We will only define the resultant of a a quadratic polynomial and a linear polynomial as this is the case relevant to our work.\footnote{For the general definition of Resultant, see Definition 2 in $\mathsection 5$ of Chapter 3 in \cite{CLO}.} Let $A,B\in\CRing{x}{n}$. View $A$ and $B$ as polynomials in $x_1$ over $\C[x_2,\ldots,x_n]$ and assume that $\deg_{x_1}(A) = 2 $ and  $\deg_{x_1}(B) = 1 $, namely,
$$A = \alpha x_1^2 + ax_1 + A_0 \quad \text{ and } B =  b x_1 + B_0 \;.$$
Then, the resultant of $A$ and $B$ with respect to $x_1$ is the determinant of their Sylvester matrix
$$
\res_{x_1}(A,B) \eqdef 
\left| \begin{bmatrix} A_0 & B_0 & 0 \\
a & b & B_0 \\
\alpha & 0 & b \\

\end{bmatrix}\right| \;.
$$
A useful fact is that if the resultant of $A$ and $B$ vanishes then they share a common factor. 

\begin{theorem}[See e.g. Proposition 8 in $\mathsection 5$ of Chapter 3 in \cite{CLO}]\label{thm:res}
Given $F,G\in\F[\MVar{x}{n}]$ of positive degree in $x_1$, the resultant $\res_{x_1}(F,G)$ is an integer polynomial in the coefficients of $F$ and $G$. Furthermore, $F$ and $G$ have a common factor in $\F[\MVar{x}{n}]$ if and only if $\res_{x_1}(F,G) = 0$.
\end{theorem}

\subsection{Rank of Quadratic Polynomials}\label{sec:rank}

%

In this section we define the rank of a quadratic polynomial, and present some of its useful properties.

\begin{definition}\label{def:rank-s}
For a homogeneous quadratic polynomial  $Q$ we denote with $\rank_s(Q)$ the minimal $r$ such that there are $2r$  linear forms $\{a_k\}_{k=1}^{2r}$ satisfying $Q=\sum_{k=1}^r a_{2k}\cdot a_{2k-1}$. We call such representation a \emph{minimal representation} of $Q$.
\end{definition}

This is a slightly different definition than the usual way one defines rank of quadratic forms,\footnote{$\rank(Q)$ is the minimal $t$ such that there are $t$ linear forms $\{a_k\}_{k=1}^{t}$, satisfying $Q=\sum_{k=1}^t a_k^2$.}
but it is more suitable for our needs. We note that a quadratic $Q$ is irreducible if and only if $\rank_s(Q)>1$. The next claim shows that a minimal representation is unique in the sense that the
space spanned by the linear forms in it is unique.

\begin{claim}\label{cla:irr-quad-span}
Let $Q$ be a homogeneous quadratic polynomial and let $Q=\sum_{i=1}^{r}a_{2i-1}\cdot a_{2i}$ and $Q = \sum_{i=1}^{r}b_{2i-1}\cdot b_{2i}$ be two different minimal representations of $Q$. Then $\spn{\MVar{a}{2r}} =\spn{\MVar{b}{2r}}$.
\end{claim}

\begin{proof}
Note that if the statement does not hold then,  without loss of generality, $a_1$ is not contained in the span of the $b_i$'s. This means that when setting $a_1=0$ the $b_i$'s are not affected on the one hand, thus $Q$ remains the same function of the $b_i$'s, and in particular $\rank_s(Q|_{a_1=0})=r$, but on the other hand $\rank_s(Q|_{a_1=0})=r-1$ (when considering its representation with the $a_i$'s), in contradiction.
\end{proof}

This claim allows us to define the notion of \textit{minimal space} of a quadratic polynomial $Q$, which we shall denote $\MS(Q)$.
\begin{definition}\label{def:MS}
Let Q be a quadratic polynomial, where $\rank_s(Q) = r$, and let $Q = \sum \limits_{i=1}^r a_{2i-1}a_{2i}$ be some minimal representation of $Q$.
Define $\MS(Q)\eqdef \spn{\MVar{a}{2r}}$, also denote $\MS(\MVar{Q}{k}) = \sum \limits_{i=1}^k \MS(Q_i)$.
\end{definition}
\autoref{cla:irr-quad-span} shows that the minimal space is well defined. The following fact is easy to verify.

\begin{fact}\label{cor:containMS}
Let $Q=\sum_{i=1}^{m}a_{2i-1}\cdot a_{2i}$ be a homogeneous quadratic polynomial, then $\MS(Q)\subseteq \spn{\MVar{a}{2m}}$.
\end{fact}

We now give some basic claims regarding $\rank_s$.


\begin{claim}\label{cla:rank-mod-space}
Let $Q$ be a homogeneous quadratic polynomial with $\rank_s(Q)=r$, and let $V \subset \CRing{x}{n}$ be a linear space of linear forms such that $\dim(V)=\Delta$. Then $\rank_s(Q|_{V=0})\geq r-\Delta$.
\end{claim}
\begin{proof}
Assume  without loss of generality $V = \spn{\MVar{x}{\Delta}}$, and consider $Q\in \C[x_{\Delta+1},\ldots,x_n][\MVar{x}{\Delta}]$. There are $\MVar{a}{\Delta} \in \CRing{x}{n}$ and $Q' \in \C[x_{\Delta+1},\ldots,x_n]$ such that $Q = \sum_{i=1}^{\Delta} a_ix_i + Q'$, where $Q|_{V=0} = Q'$. As $\rank_s(\sum_{i=1}^{\Delta} a_ix_i ) \leq \Delta$, it must be that $\rank_s(Q|_{V=0}) \geq r-\Delta$.
\end{proof}

\begin{claim}\label{cla:ind-rank}
Let $P_1\in \CRing{x}{k}$, and $P_2 = y_1y_2\in \CRing{y}{2}$. Then $ \rank_s (P_1+P_2) = \rank_s(P_1) + 1$. Moreover, $y_1,y_2 \in \MS(P_1+P_2).$
\end{claim}

\begin{proof}
Denote $\rank_s(P_1) = r$ and assume towards a contradiction that there are $\MVar{a}{2r}$ linear forms in $\C[x_1,\dots,x_k,y_1,y_2]$ such that $P_1+P_2 = \sum\limits_{i=1}^{r}a_{2i-1}a_{2i}$. 
Clearly, $\sum\limits_{i=1}^{r}a_{2i-1}a_{2i} \equiv_{y_1} P_1$. As $\rank_s(P_1) = r$  this is a minimal representation of $P_1$. Hence, for every $i$, $a_i|_{y_1=0} \in\MS(P_1)\subset \CRing{x}{k}$. Moreover, from the minimality of $r$, $a_i|_{y_1=0} \neq 0$. Therefore, as $y_1$ and $y_2$ are linearly independent, we deduce that all the coefficients of $y_2$ in all the $a_i$'s are 0. By reversing the roles of $y_1$ and $y_2$ we can conclude that $\MVar{a}{2r}\subset \CRing{x}{k}$ which  means that $Q$ does not depend on $y_1$ and $y_2$ in contradiction. Consider a minimal representation $P_1 = \sum_{i=1}^{2r} b_{2i-1}b_{2i}$,  from the fact that $\rank_s(P_1+P_2) = r+1$ it follows that  $P_1+P_2 = \sum_{i=1}^{2r} b_{2i-1}b_{2i}+y_1y_2$ is a minimal representation of $P_1+P_2$ and thus $\MS(P_1+P_2) =\MS(P_1) + \spn{y_1,y_2}$. 
\end{proof}

\begin{corollary}\label{cla:intersection}
Let $a$ and $b$ be linearly independent linear forms. Then, if $c,d,e$ and $f$ are linear forms such that $ab+cd=ef$ then $\dim(\spn{a,b}\cap \spn{c,d})\geq 1$.
\end{corollary}

\begin{claim}\label{cla:rank-2-in-V}
Let $a,b,c$ and $d$ be linear forms, and $V$ be a linear space of linear forms. Assume $\{0\} \neq \MS(ab-cd) \subseteq V$ then $\spn{a,b} \cap V\neq \{0\}$.	
\end{claim}

\begin{proof}
	As $\MS(ab-cd) \subseteq V$ it follows that $ab \equiv_{V} cd$. If both sides are zero then $ab\in \ideal{V}$ and without loss of generality $b\in V$ and the statement holds. If neither sides is zero then from \autoref{fact:ufd} there are linear forms $v_1, v_2 \in V$, and $\lambda_1, \lambda_2 \in \C^\times$ such that, $\lambda_1\lambda_2 = 1$ and without loss of generality $c = \lambda_1a + v_1, d= \lambda_2 b+v_2$. Note that not both $v_1,v_2$ are zero, as $ab-cd \neq 0$. Thus, \[ab-cd = ab - (\lambda_1 a +  v_1)(\lambda_2 b+v_2)= \lambda_1 av_2+\lambda_2 bv_1 + v_1v_2. \]
	As $\MS(ab-cd) \subseteq V$ it follows that $\MS(\lambda_1 av_2+\lambda_2 bv_1) \subseteq V$ and therefore there is a linear combination of $a,b$ in $V$ and the statement holds.
\end{proof}

We end this section with claims that will be useful in our proofs. 
 
\begin{claim}\label{cla:linear-spaces-intersaction }
Let $V = \sum_{i=1}^m V_i$ where $V_i$ are linear subspaces, and for every $i$, $\dim(V_i) = 2$. If for every $i\neq j \in [m]$, $\dim(V_i\cap V_j) = 1$, then either $\dim(\bigcap_{i=1}^m V_i) = 1$ or $\dim(V)=3$. 
\end{claim}
\begin{proof}
Let  $w \in V_1\cap V_2$. Complete it to basis of $V_1$ and $V_2$: $V_1 = \spn{u_1,w}$ and  $V_2 = \spn{u_2,w}$. Assume that $\dim(\bigcap_{i=1}^m V_i) = 0$. Then, there is some $i$ for which $w \notin V_i$.
Let
$x_1 \in V_i\cap V_1$, and so $x_1 = \alpha_1 u_1 + \beta_1 w$, where $\alpha_1 \neq 0$. Similarly, let $x_2 \in V_i\cap V_2$. Since $w \notin V_i$, $x_2 = \alpha_2 u_2 + \beta_2 w$, where $\alpha_2 \neq 0$.  Note that $x_1 \notin \spn{ x_2}$, as $\dim(V_1\cap V_2) = 1$, and $w$ is already in their intersection. Thus, we have $V_i = \spn{x_1,x_2} \subset \spn{w, u_1, u_2}$.

Now, consider any other $j \in [m]$. If $V_j$ does not contain $w$, we can apply the same argument as we did for $V_i$ and conclude that $V_j \subset \spn{w,u_1,u_2}$. On the other hand, if $w \in V_j$, then let $x_j \in V_i\cap V_j$, it is easy to see that $x_j, w$ are linearly independent and so $V_j = \spn{w,x_j} \subset \spn{w,V_i} \subseteq \spn{w, u_1,u_2}$. Thus, in any case $V_j \subset \spn{w, u_1,u_2}$. In particular, $\sum_j V_j \subseteq \spn{w, u_1,u_2}$ as claimed.
\end{proof}
\subsection{Projection Mappings}\label{sec:z-map}

In this section we present and apply a new technique which allows us to simplify the structure of quadratic polynomials.
Naively, when we want to simplify a polynomial equation, we can project it on a subset of the variables.
Unfortunately, this projection does not necessarily preserve pairwise linear independence, which is a crucial property in our proofs. 
To remedy this fact, we present a set of mappings, which are somewhat similar to projections, but do preserve pairwise linear independence among polynomials.
\begin{definition}\label{def:z-mapping}
Let $V = \spn{\MVar{v}{\Delta}}\subseteq \spn{x_1,\ldots,x_n}$ be a $\Delta$-dimensional linear space of linear forms, and let $\{\MVar{u}{{n-\Delta}}\}$ be a basis for $V^\perp$. For $\vaa = (\MVar{\alpha}{\Delta})\in \C^{\Delta}$ we define $T_{\vaa, V} : \CRing{x}{n} \mapsto \C[\MVar{x}{n},z]$, where $z$ is a new variable, to be the linear map given by the following action on the basis vectors: $T_{\vaa, V}(v_i) = \alpha_i z$ and $T_{\vaa, V}(u_i)=u_i$.
\end{definition}

\begin{observation}
$T_{\vaa, V}$ is a linear transformation and is also a ring homomorphism. This follows from the fact that a basis for $\spn{\MVar{x}{n}}$ is a basis for $\CRing{x}{n}$ as $\C$-algebra. 
\end{observation}

\begin{claim}\label{cla:res-z-ampping}
	Let $V\subseteq \spn{x_1,\ldots,x_n}$ be a $\Delta$-dimensional linear space of linear forms. Let $F$ and  $G$ be two polynomials that share no common irreducible factor. Then, with probability $1$ over the choice of $\vaa \in [0,1]^{\Delta}$ (say according to the uniform distribution), $T_{\vaa, V}(F)$ and $T_{\vaa, V}(G)$ do not share a common factor that is not a polynomial in $z$.
\end{claim}

\begin{proof}
	Let $\lbrace\MVar{u}{{n-\Delta}}\rbrace$ be a basis for $V^\perp$. We think of $F$ and $G$ as polynomials in $\C[\MVar{v}{\Delta}, \MVar{u}{n-\Delta}]$.
	As $T_{\vaa,V}: \C[\MVar{v}{\Delta}, \MVar{u}{n-\Delta}] \rightarrow \C[z, \MVar{u}{n-\Delta}]$,  
	\autoref{thm:res} implies that if $T_{\vaa, V}(F)$ and  $T_{\vaa, V}(G)$ share a common factor that is not a polynomial in $z$, then, without loss of generality, their resultant with respect to $u_1$ is zero. \autoref{thm:res} also implies that the resultant of $F$ and $G$ with respect to $u_1$ is not zero. Observe that with probability $1$ over the choice of $\vaa$, we have that $\deg_{u_1}(F) = \deg_{u_1}(T_{\vaa, V}(F))$ and $\deg_{u_1}(G) = \deg_{u_1}(T_{\vaa, V}(G))$.
	As $T_{\vaa, V}$ is a ring homomorphism this implies that $Res_{u_1}(T_{\vaa, V}(G), T_{\vaa, V}(F)) = T_{\vaa, V}(Res_{u_1}(G, F))$.  The Schwartz-Zippel-DeMillo-Lipton lemma now implies that sending each basis element of $V$ to a random multiple of $z$, chosen uniformly from $(0,1)$ will keep the resultant non zero with probability $1$. This also means that $T_{\vaa, V}(F)$ and $T_{\vaa, V}(G)$ share no common factor.
\end{proof}

\begin{corollary}\label{cla:still-indep}
	Let $V$ be a $\Delta$-dimensional linear space of linear forms. Let $F$ and  $G$ be two linearly independent, irreducible quadratics, such that $\MS(F),\MS(G)\not\subseteq V$. Then, with probability $1$ over the choice of $\vaa \in [0,1]^{\Delta}$ (say according to the uniform distribution), $T_{\vaa, V}(F)$ and $T_{\vaa, V}(G)$ are linearly independent.
\end{corollary}

\begin{proof}
	As $F$ and $G$ are irreducible  they share no common factors. \autoref{cla:res-z-ampping} implies that  $T_{\vaa, V}(F)$ and $T_{\vaa, V}(G)$ do not share a common factor that is not a polynomial in $z$.  The Schwartz-Zippel-DeMillo-Lipton implies that with probability $1$,   $T_{\vaa, V}(F)$ and $T_{\vaa, V}(G)$ are not polynomials in $z$, and therefore they are linearly independent.
\end{proof}

\begin{claim}\label{cla:z-map-rank}
Let $Q$ be an irreducible quadratic polynomial, and $V$  a $\Delta$-dimensional linear space.
Then for every $\vaa \in \C^{\Delta}$,  $\rank_s(T_{\vaa, V}(Q)) \geq \rank_s(Q)-\Delta$.
\end{claim}
\begin{proof}
$\rank_s(T_{\vaa, V}(Q)) \geq \rank_s(T_{\vaa, V}(Q)|_{z=0}) = \rank_s(Q|_{V=0}) \geq \rank_s(Q)-\Delta$, where the last inequality follows from \autoref{cla:rank-mod-space}.
\end{proof}

\begin{claim}\label{cla:z-map-dimension}
Let $\cQ$ be a set of quadratics, and $V$ be a $\Delta$-dimensional linear space. Then, if there are linearly independent vectors, $\{\vaa^1,\dots, \vaa^{\Delta}  \}\subset \C^{\Delta}$, such that, for every $i$,\footnote{Recall that $\MS(T_{\vaa^i,V}(\cQ))$ is the space spanned by $\cup_{Q\in\cQ}\MS(T_{\vaa^i,V}(\cQ))$.} $\dim(\MS(T_{\vaa^i,V}(\cQ)))\leq \sigma$ then $\dim(\MS(\cQ))\leq (\sigma+1) \Delta$.
\end{claim}
\begin{proof}
As $\dim(\MS(T_{\vaa^i,V}(\cQ)))\leq \sigma$, there are $\MVar{u^i}{\sigma} \subset V^{\perp}$ such that $\MS(T_{\vaa^i,V}(\cQ)) \subseteq \spn{z,\MVar{u^i}{\sigma} }$. We will show that $\MS(\cQ) \subset V + \spn{\{\MVar{u^i}{\sigma}\}_{i=1}^{\Delta}}$, which is of dimension at most $\Delta+ \sigma\Delta$.

Let $P \in \cQ$, then there are linear forms, $\MVar{a}{\Delta} \subset V^{\perp}$ and polynomials $P_{V} \in  \C[V]$ and $P'\in \C[V^{\perp}]$, such that
\begin{equation*}
P = P_{V}  + \sum\limits_{j=1}^{\Delta} a_jv_j + P'.
\end{equation*}
Therefore, after taking the projection for a specific $T_{\vaa^i, V}$, for some $\gamma \in \C$,
\begin{equation*}
T_{\vaa^i,V}(P) = \gamma z^2 + \left(\sum\limits_{j=1}^{\Delta} \alpha^i_ja_j\right) z + P'.
\end{equation*}
Denote $b_{P,i} = \sum\limits_{j=1}^{\Delta} \alpha^i_ja_j$. By \autoref{cla:still-indep} if $\MVar{a}{\Delta}$ are not all zeros, then, with probability $1$, $b_{P,i} \neq \vec{0}$ .

If $b_{P,i}\notin \MS(P')$ then from \autoref{cla:ind-rank} it follows that  $\lbrace z, b_{P,i}, \MS(P')\rbrace \subseteq \spn{ \MS(T_{\vaa^i,V}(P))}$. If, on the other hand, $b_{P,i} \in \MS(P')$, then clearly $\lbrace b_{P,i}, \MS(P')\rbrace \subseteq \spn{ z, \MS(T_{\vaa^i,V}(P))}$. To conclude, in either case, $\{b_{P,i},\MS(P')\} \subseteq \spn{z, \MVar{u^i}{\sigma}}.$

Applying the analysis above to $T_{\vaa^1,V},\dots,T_{\vaa^{\Delta},V}$ we obtain that $\spn{b_{P,1}, \cdots b_{P,\Delta}}\subseteq \spn{\{\MVar{u^i}{\sigma}\}_{i=1}^{\Delta}}$. As $\vaa^1,\dots\vaa^\Delta$ are linearly independent, we have that $\{{\MVar{a}{\Delta}}\} \subset \spn{b_{P,1}, \cdots b_{P,\Delta}}$, and thus $\MS(P) \subseteq V + \lbrace\MVar{a}{\Delta}\rbrace + LS(P') \subseteq V + \spn{\{\MVar{u^i}{\sigma}\}_{i=1}^{\Delta}}$.
\end{proof}

\section{Structure theorem for quadratics  satisfying $\prod_i Q_i\in\sqrt{(A,B)}$}\label{sec:structure}


An important tool in the proofs of our main results is \autoref{thm:structure-intro} that classifies all the possible cases in which a product of  quadratic polynomials $Q_1\cdot Q_2\cdots Q_k$ is in the radical of two other quadratics, $\sqrt{\ideal{A,B}}$. To ease the reading we repeat the statement of the theorem here, albeit with slightly different notation.

\begin{theorem}\label{thm:structure}
Let $\{Q_k\}_{k\in \cK},A,B$ be homogeneous polynomials of degree $2$ such that $\prod_{k\in \cK}Q_k \in \sqrt{\ideal{A,B}}$. Then one of the following cases hold:
\begin{enumerate}[label={(\roman*)}]
\item There is $k\in \cK$ such that $Q_k$ is in the linear span of $A,B$  \label{case:span}

\item \label{case:rk1}
There exists a non trivial linear combination of the form $\alpha A+\beta B = c\cdot d$ where  $c$ and $d$ are linear forms.
 
\item There exist two linear forms $c$ and $d$ such that when setting $c=d=0$ we get that $A,B$ and one of $\{Q_k\}_{k\in \cK}$ vanish. \label{case:2}
\end{enumerate}
\end{theorem}

From now on, to ease notations, we use \autoref{thm:structure}\ref{case:span}, \autoref{thm:structure}\ref{case:rk1} or \autoref{thm:structure}\ref{case:2} to describe different cases of \autoref{thm:structure}.

The following claim of \cite{Gupta14} shows that we can assume $|\cK| = 4$ in the statement of \autoref{thm:structure}.

\begin{claim}[Claim 11 in \cite{Gupta14}]\label{cla:gup-4}
Let $\MVar{P}{d},\MVar{Q}{k}\in \CRing{x}{n}$ be homogeneous and the degree of each $P_i$ is at most $r$. Then,
\begin{equation*}
\prod_{i=1}^k Q_i \in \sqrt{\ideal{\MVar{P}{d}}} \Rightarrow \exists \lbrace \MVar{i}{r^d} \rbrace \subset [k] \; \text{ such that} \quad  \prod_{j=1}^{r^d} Q_{i_j} \in \sqrt{\ideal{\MVar{P}{d}}} \;.
\end{equation*}
\end{claim}
Thus, for $r=d=2$ it follow that there are at most four polynomials among the $Q_i$s whose product is in $\sqrt{\ideal{A,B}}$. \\


Before proving  \autoref{thm:structure} we explain the intuition behind the different cases in the theorem. Clearly, if one of $Q_1,\ldots,Q_4$ is a linear combination of $A,B$ then it is in their radical (and in fact, in their linear span). If $A$ and $B$ span a product of the form $ab$ then, say, $(A+ac)(A+bd)$ is in their radical. Indeed, $\sqrt{\ideal{A,B}} = \sqrt{\ideal{A,ab}}$. This case is clearly different than the linear span case. Finally, we note that if $A=ac+bd$ and $B=ae+bf$ then the product $a\cdot b\cdot (cf-de)$ is in $\sqrt{\ideal{A,B}}$.\footnote{If we insist on having all factors of degree $2$ then the same argument shows that  the product $(a^2+A)\cdot (b^2+B)\cdot (cf-de)$ is in $\sqrt{\ideal{A,B}}$.} This case is different than the other two cases as $A$ and $B$ do not span any linear form (or any reducible quadratic) non trivially.


Thus, all the three cases are distinct and can happen. What \autoref{thm:structure} shows is that, essentially, these are the only possible cases.

\begin{proof}[Proof of \autoref{thm:structure}]
Following \autoref{cla:gup-4} shall assume in the proof that $|\cK|=4$.
By applying a suitable linear transformation we can assume that for some $r\geq 1$  $$A = \sum_{i=1}^{r} x_i^2.$$

%

We can also assume  without loss of generality that $x_1^2$ appears only in $A$ as we can replace $B$ with any polynomial of the form $B'=B - \alpha A$ without affecting the result as $\ideal{A,B} = \ideal{A,B'}$. Furthermore, all cases in the theorem remain the same if we replace $B$ with $B'$ and vice versa.

In a similar fashion we can replace $Q_1$ with $Q_1'=Q_1-\alpha A$ to get rid of the term $x_1^2$ in $Q_1$. We can do the same for the other $Q_i$s. Thus, without loss of generality, the situation is 
\begin{eqnarray}
A &=& x_1^2 - A' \nonumber \\
B &=& x_1\cdot b - B'\ \label{eq:Q2}\\
Q_i  &=& x_1 \cdot b_i - Q'_i  \quad \text{for} \quad  i\in\{1,2,3,4\}  \nonumber 
\end{eqnarray}
where  $A',b,B',Q'_i,b_i$ are homogeneous polynomials that do not depend on $x_1$. 
The analysis shall deal with two cases according to whether $B$ depends on $x_1$ or not, as we  only consider the resultant of $A$ and $B$ with respect to $x_1$ when it appears in both polynomials. 

\paragraph{Case $b\not\equiv 0$:}
Consider the Resultant of $A$ and $B$  with respect to $x_1$. It is easy to see that
$$\res_{x_1}(A,B) = {B'}^2 - b^2 \cdot A'.$$

We first prove that if the resultant is irreducible then Case~\ref{case:span} of \autoref{thm:structure} holds. For this we shall need the following claim.

\begin{claim}\label{cla:res-vanish}
Whenever $\res_{x_1}(A,B) =0$ it holds that $\prod_{i=1}^{4}(B'\cdot b_i - b\cdot Q'_i) =0$. 
\end{claim}

\begin{proof}
Let $\vaa \in \C^{n-1}$ be such that $\res_{x_1}(A,B)(\vaa)=0$ then either $b(\vaa)=0$, which also implies $B'(\vaa)=0$ and in this case the claim clearly holds,  or  $b(\vaa)\neq 0$. Consider the case $b(\vaa)\neq 0$ and set $x_1=B'(\vaa)/b(\vaa)$ (we are free to select a value for $x_1$ as $\res_{x_1}(A,B)$ does not involve $x_1$). Notice that for this substitution we have that $B(\vaa)=0$ and that 
$$A|_{x_1=B'(\vaa)/b(\vaa)} =  (B'(\vaa)/b(\vaa))^2 -  A'(\vaa) = \res_{x_1}(A,B)(\vaa)/b(\vaa)^2=0.$$
Hence, we also have $\prod_{i=1}^{4}Q_i |_{x_1=B'(\vaa)/b(\vaa)} = 0$. In other words that 
$$\left(\frac{1}{b^4}\prod_{i=1}^{4}(B'\cdot b_i - b\cdot Q'_i)\right)(\vaa)=0 \;.$$
\end{proof}
It follows that 
$$\prod_{i=1}^{4}(B'\cdot b_i - b\cdot Q'_i) \in \sqrt{\res_{x_1}(A,B)} \;.$$
In other words, for some positive integer $k$ we have that $\res_{x_1}(A,B)$ divides $\left(\prod_{i=1}^{4}(B'\cdot b_i - b\cdot Q'_i)\right)^k$. 
As every irreducible factor of $\left(\prod_{i=1}^{4}(B'\cdot b_i - b\cdot Q'_i)\right)^k$ is of degree $3$ or less, we get that if the resultant is irreducible then one of the multiplicands must be identically zero. Assume  without loss of generality that $B'b_1-bB'_1=0$. It is not hard to verify that in this case either $Q_1$ is a scalar multiple of $B$  and then \autoref{thm:structure}\ref{case:span} holds, or that $B'$ is divisible by $b$. However, in the latter case it also holds that $b$ divides the resultant, contradicting the assumption that it is irreducible.

From now on we assume that  $\res_{x_1}(A,B)$ is reducible. We consider two possibilities. Either $\res_{x_1}(A,B)$ has a linear factor or it can be written as 
$$\res_{x_1}(Q_1,Q_2)=C\cdot D,$$ for irreducible quadratic polynomials $C$ and $D$.

Consider the case where the resultant has a linear factor. If that linear factor is $b$ then $b$ also divides $B$ and \autoref{thm:structure}\ref{case:rk1} holds. Otherwise, if it is a different linear form $\ell$ then when setting $\ell=0$ we get that the resultant of $A|_{\ell=0}$ and $B|_{\ell=0}$ is zero and hence either $B|_{\ell=0}$ is identically zero and \autoref{thm:structure}\ref{case:rk1} holds, or they share a common factor (see \autoref{thm:res}).  It is not hard to see that if that common factor is of degree $2$ then \autoref{thm:structure}\ref{case:rk1} holds and if it is a linear factor then \autoref{thm:structure}\ref{case:2} holds.

Thus, the only case left to handle (when $b\not \equiv 0$) is when there are two irreducible quadratic polynomials, $C$ and $D$ such that $CD=\res_{x_1}(A,B)$. As $C$ and $D$ divide two multiplicands in $\prod_{i=1}^{4}(B'\cdot b_i - b\cdot Q'_i) $ we can assume, without loss of generality, that 
$(B'\cdot b_3 - b\cdot Q'_3)\cdot (B'\cdot b_4 - b\cdot Q'_4) \in \sqrt{\ideal{\res_{x_1}(A,B)}} $.
Next, we express $A',B',C$ and $D$ as quadratics over $b$. That is
\begin{eqnarray}
A' &=& \alpha b^2 + a_1 b + A''
\label{eq:y} \\
B' &=& \beta b^2 + a_2 b + B''\nonumber \\
C &=& \gamma b^2 + a_3 b + C'' \nonumber  \\
D &=& \delta b^2 + a_4 b + D'',\nonumber 
\end{eqnarray}
where $a_1,\ldots,D''$ do not involve $b$ (nor $x_1$).
We have the following two representations of the resultant:
\begin{eqnarray}
\res_{x_1}(A,B)&=&{B'}^2 - b^2 \cdot A'  \label{eq:res-q}\\
&=& \beta^2\cdot b^4 + 2\beta a_2 \cdot b^3 + (2\beta B''  + a_2^2) \cdot b^2 + 2a_2B'' \cdot b + {B''}^2 - \alpha b^4- a_1 b^3 - A'' b^2 \nonumber \\
&=& (\beta^2-\alpha)b^4 + (2\beta a_2 - a_1)\cdot b^3 + (2\beta B'' + a_2^2- A'') \cdot b^2 + 2a_2B'' \cdot b + {B''}^2 \nonumber
\end{eqnarray}
and
\begin{eqnarray}
\res_{x_1}(Q_1,Q_2)&=&CD\label{eq:res-bc} \\
&=& (\gamma b^2 + a_3 b + C'')\cdot (\delta b^2 + a_4 b + D'')\nonumber \\
&=& \gamma\delta b^4 + (\gamma a_4 + \delta a_3)b^3 + (\gamma D'' + a_3 a_4 + \delta C'')b^2 + (a_3D'' + a_4C'')b + C''D'' .\nonumber 
\end{eqnarray}
Comparing the different coefficients of $b$ in the two representations in Equations~\ref{eq:res-q} and \ref{eq:res-bc} we obtain the following equalities
\begin{eqnarray}
{B''}^2 &=&  C''D'' \label{eq:coeff:1}\\
2a_2B'' &=& a_3D'' + a_4C'' \label{eq:coeff:y}
\end{eqnarray}
We now consider the two possible cases giving \autoref{eq:coeff:1}.
\begin{enumerate}
\item {\bf Case 1  explaining \autoref{eq:coeff:1}:} After rescaling $C$ and $D$ we have that ${B''}=C"=D''$. \autoref{eq:y} implies that for some linear form $u,v$ we have that
$$C = bv + B' \quad \text{and} \quad D = bu+B'\;.$$
We now expand the resultant again:
\begin{eqnarray*}
{B'}^2 + b(v+u)B' + b^2 vu=(bv + B')\cdot ( bu+B') &=& CD \\ 
&=& \res_{x_1}(A,B)= {B'}^2 - b^2 A'
\end{eqnarray*}
Hence, 
\begin{equation}\label{eq:y-cor}
(v+u)B' + b vu=- b A' \;.
\end{equation}
Thus, either $b$ divides $B'$ in which case we get that $b$ divides $B$ and we are done as \autoref{thm:structure}~\ref{case:rk1} holds, or $b$ divides $u+v$. That is, 
\begin{equation}\label{eq:u+v}
u+v = \epsilon b
\end{equation} 
for some constant $\epsilon\in \C$. Plugging this back into \autoref{eq:y-cor} we get 
$$\epsilon b B' + b vu=- b A' \;.$$
In other words, 
$$\epsilon B' + vu= -  A' \;.$$
Consider the linear combination $Q=A + \epsilon B$. 
We get that
\begin{eqnarray}
Q=A + \epsilon B &=& (x_1^2 - A') + \epsilon (x_1 b-B') \nonumber \\ 
&=& x_1^2 + \epsilon x_1 b + vu  \nonumber \\
&=& x_1^2 + x_1(u+v)+uv  \nonumber\\
&=& (x_1+u)(x_1+v)\;.\label{eq:Q-reduce}
\end{eqnarray}
where the equality in the third line follows from \autoref{eq:u+v}. Thus, \autoref{eq:Q-reduce} shows that some linear combination of $A$ and $B$ is reducible which implies that \autoref{thm:structure}\ref{case:rk1} holds.

\item {\bf Case 2 explaining \autoref{eq:coeff:1}:} ${B''} = u\cdot v$ and we have that,  without loss of generality, $C''=u^2$ and $D''=v^2$ (where $u,v$ are linear forms).
Consider \autoref{eq:coeff:y}. We have that $v$ divides $2a_2B'' - a_3D''$. It follows that $v$ is also a factor of $ a_4C'' $. Thus, either $u$ is a multiple of $v$ and we are back in the case where $C''$ and $D''$ are multiples of each other, or $a_4$ is a multiple of $v$. In this case we get from \autoref{eq:y} that for some constant $\delta'$,  
$$D = \delta b^2 + a_4 b + D'' = \delta b^2 +\delta' v b + v^2.$$
Thus, $D$ is a homogeneous polynomial in two linear forms. Hence, $D$ is reducible, in contradiction. 
\end{enumerate}
This concludes the proof of \autoref{thm:structure} for the case $b\not\equiv 0$.

\paragraph{Case $b\equiv 0$:}

To ease notation let use denote $x=x_1$. We have that
$A = x^2 - A' $ and that $x$ does not appear in $A',B$. Let $y$ be some variable such that $B=y^2-B'$, and $B'$ does not involve $y$ (we can always assume this is the case without loss of generality). As before we can subtract a multiple of $B$ from $A$ so that the term $y^2$ does not appear   in $A$. If $A$ still involves $y$ then we are back in the previous case (treating $y$ as the variable according to which we take the resultant). Thus, the only case left to study is when there are two variables $x$ and $y$ such that
$$A = x^2 - A' \quad \text{and} \quad B= y^2 - B' \;,$$ where neither $A'$ nor $B'$ involve either $x$ or $y$. To ease notation denote the rest of the variables as $\vz$. Thus, $A'=A'(\vz)$ and $B'=B'(\vz)$. It is immediate that for any assignment to $\vz$ there is an assignment to $x,y$ that yields a common zero of $A,B$.

By subtracting linear combinations of $A$ and $B$ from the $Q_i$s we can assume that for every $i\in [4]$
$$Q_i = \alpha_i xy + a_i(\vz) x + b_i(\vz) y + Q'_i(\vz) \;.$$


We next show that, under the assumptions in the theorem statement it must be the case that either $A'$ or $B'$ is a perfect square or that $A'\sim B'$. In either situation we have that \autoref{thm:structure}\ref{case:rk1} holds.
We next show that if $A'$ and $B'$ are  linearly independent then this  implies that   at least  one of $A',B'$ is a perfect square. 

Let $Z(A,B)$ be the set of common zeros of $A$ and $B$, and denote by $\pi_{\vz}: Z(A,B) \rightarrow \C^{n-2}$, the projection on the $\vz$ coordinates. Note that $\pi_{\vz}$ is surjective, as for any assignment to $\vz$ there is an assignment to $x,y$ that yields a common zero of $A,B$. 
\begin{claim}
Let $Z(A,B) = \bigcup_{i=1}^k X_k$, be the decomposition of $Z(A,B)$ to irreducible components. Then there exists $i\in [k]$ such that $\pi_{\vz}(X_i)$ is dense in $\C^{n-2}$.
\end{claim}
\begin{proof}
 $\bigcup_{i=1}^k \pi_{\vz}( X_i) = \pi_{\vz}(Z(A,B)) = \C^{n-2}$, as $\pi_{\vz}$ is a surjection, it holds that  $\bigcup_{i=1}^k \overline{\pi_{\vz}( X_i)} = \C^{n-2}$. We also know that $\C^{n-2}$ is irreducible, and thus there is $i\in [k]$ such that $\overline{\pi_{\vz}( X_i)} = \C^{n-2}$, which implies that ${\pi_{\vz}( X_i)} $ is dense.
\end{proof}

Assume, without loss of generality that $\pi_{\vz}( X_1)$ is dense. We know that $X_1 \subseteq Z(\prod_{i=1}^{4}Q_i)$ so we can assume, without loss of generality that $X_1 \subseteq Z(Q_1)$. Observe that this implies that $Q_1$ must depend on at least one of $x,y$. Indeed, if $Q_1$ depends on neither then it is a polynomial in $\vz$ and hence its set of zeros cannot be dense. 

Every point $\vxx \in X_1$ is of the form $\vxx = (\delta_1\sqrt{A'(\vbb)},\delta_2\sqrt{B'(\vbb)},\vbb)$, for some $\vbb \in \C^{n-2}$, $\delta_1,\delta_2 \in \{\pm 1\}$ ($ \delta_1,\delta_2$ may be a function of $\vbb$).  Thus $Q_1(\vxx) = Q_1(\delta_1\sqrt{A'(\vbb)},\delta_2\sqrt{B'(\vbb)},\vbb) =0$, and we obtain that 

\begin{equation}\label{eq:sqrt}
\alpha_1 \delta_1\delta_2\sqrt{A'(\vbb')}\cdot \sqrt{B'(\vbb')} + a_1(\vbb') \delta_1 \sqrt{A'(\vbb')} + b_1(\vbb') \delta_2\sqrt{B'(\vbb')} + Q'_1(\vbb')=0\;.
\end{equation}
As we assumed that $Q_1$ depends on at least one of $x,y$ let us assume  without loss of generality that either $\alpha_1$ or $a_1$ are non zero. The next argument is similar to the proof that $\sqrt{2}$ is irrational. Note that we use the fact that $\delta_1^2=\delta_2^2=1$.

\begin{eqnarray}
\eqref{eq:sqrt} 
\implies &  B'(\vbb')\left( \alpha_1 \delta_1 \sqrt{A'(\vbb')} + b_1(\vbb') \right)^2 = \left( Q'_1(\vbb') +a_1(\vbb') \delta_1 \sqrt{A'(\vbb')} \right)^2 \nonumber \\
\implies & B'(\vbb')\left( \alpha_1^2 {A'(\vbb')} +2\delta_1\alpha_1 b_1(\vbb') \sqrt{A'(\vbb')} + 
 b_1(\vbb')^2 \right) = \nonumber \\ & \quad\quad\quad Q'_1(\vbb')^2 + 2\delta_1 a_1(\vbb')Q'_1(\vbb') \sqrt{A'(\vbb')} + a_1(\vbb')^2 {A'(\vbb')} \nonumber \\
 \implies &\delta_1 \sqrt{A'(\vbb')}\left( 2\alpha_1 b_1(\vbb') B'(\vbb') -  2a_1(\vbb')Q'_1(\vbb') \right) = \label{eq:sqrt-2} \\
 & \quad\quad\quad Q'_1(\vbb')^2  +a_1(\vbb')^2 {A'(\vbb')} -B'(\vbb')\left( \alpha_1^2 {A'(\vbb')}  + 
 b_1(\vbb')^2 \right)\nonumber  \\
 \implies & A'(\vbb') \left( 2\alpha_1 b_1(\vbb') B'(\vbb') -  2a_1(\vbb')Q'_1(\vbb') \right)^2 = \label{eq:C=0} \\
 & \quad\quad\quad \left(Q'_1(\vbb')^2  +a_1(\vbb')^2 {A'(\vbb')} -B'(\vbb')\left( \alpha_1^2 {A'(\vbb')}  + 
 b_1(\vbb')^2 \right)\right)^2 \nonumber  \;.
\end{eqnarray}
This equality holds for every $\vbb \in \pi_{\vz}(X_1)$, which is a dense set, and hence holds as a polynomial identity.
Thus, either $A'(\vz)$ is a square, in which case we are done or it must be the case that the following identities hold
\begin{equation}
Q'_1(\vz)^2  +a_1(\vz)^2 {A'(\vz)} -B'(\vz)\left( \alpha_1^2 {A'(\vz)}  + 
 b_1(\vz)^2 \right) = 0 \label{eq:C=01}
 \end{equation}
 and
 \begin{equation}
 \alpha_1 b_1(\vz) B'(\vz) -  a_1(\vz)Q'_1(\vz) =0 \;.\label{eq:C=02}
\end{equation}
By symmetry, if $B'(\vz)$ is not a square (as otherwise we are done), we get that 
 \begin{equation}
 \alpha_1 a_1(\vz) A'(\vz) -  b_1(\vz)Q'_1(\vz) =0 \;.\label{eq:C=03}
\end{equation}
If $\alpha_1=0$ then we get from \eqref{eq:C=02} that $Q'_1\equiv 0$. 
Hence,
by \eqref{eq:C=01},
$$a_1(\vz)^2 {A'(\vz)} = B'(\vz)b_1(\vz)^2\;.$$
Since we assumed that $A'$ and $B'$ are independent this implies that $A'$ and $B'$ are both squares.
If $Q'_1\not\equiv 0$ (and in particular, $\alpha_1\neq 0$ ) then either $a_1(\vz)=b_1(\vz)\equiv 0$, in which case Equation~\eqref{eq:C=01} implies that $Q'_1(\vz)^2 =\alpha_1^2 A'(\vz)B'(\vz)$ and we are done (as either both $A'$ and $B'$ are squares or they are both multiples of $Q'_1$), or Equations~\eqref{eq:C=02},\eqref{eq:C=03} imply that $\alpha_1^2  A'(\vz) B'(\vz) = Q'_1(\vz)^2$
which again implies the claim.

This concludes the proof of  \autoref{thm:structure} for the case $b\equiv 0$ and thus the proof of the theorem.
\end{proof}

\section{Sylvester-Gallai theorem for quadratic polynomials}\label{sec:quad-SG}

In this section we prove \autoref{thm:main-sg-intro}. For convenience we repeat the statement of the theorem.




\begin{theorem*}[\autoref{thm:main-sg-intro}]
There exists a universal constant $c$ such that the following holds. 
Let $\tilde{\cQ} = \{Q_i\}_{i\in \{1,\ldots,m\}}\subset\C[x_1,\ldots,x_n]$ be a finite set of pairwise linearly independent homogeneous polynomials, such that every $Q_i\in \tilde{\cQ}$ is either irreducible or a square of a linear form. Assume  that, for every $i\neq j$, whenever $Q_i$ and $Q_j$ vanish then so does $\prod_{k\in  \{1,\ldots,m\} \setminus\{i,j\}} Q_k$. Then,  $\dim(\spn{\cQ})\leq c$.
\end{theorem*}


\begin{remark}
The requirement that the polynomials are homogeneous is not essential as homogenization does not affect the property $Q_k\in\sqrt{\ideal{Q_i,Q_j}}$. 
\end{remark}

\begin{remark}
Note that we no longer demand that the polynomials are irreducible but rather allow some of them to be square of linear forms, but now we restrict all polynomials to be of degree exactly $2$. Note that both versions of the theorem are equivalent, as this modification does not affect the vanishing condition. 
\end{remark}

\begin{remark}\label{rem:K=4}
Note that from \autoref{cla:gup-4} it follows that for every $i\neq j$ there exists a subset $\cK \subseteq [m]\setminus\{i,j\}$ such that $|\cK|\leq 4$ and whenever $Q_i$ and $Q_j$ vanish then so does $\prod_{k\in \cK} Q_k$.
\end{remark}

In what follows we shall use the following terminology. Whenever we say that two quadratics $Q_1,Q_2\in\tilde\cQ$ satisfy \autoref{thm:structure}\ref{case:span} we mean that there is a polynomial $Q_3\in\tilde\cQ\setminus\{Q_1,Q_2\}$ in their linear span. Similarly, when we say that they satisfy \autoref{thm:structure}\ref{case:rk1}  (\autoref{thm:structure}\ref{case:2}) we mean that there is a reducible quadratic in their linear span (they belong to $\ideal{a_1,a_2}$ for linear forms $a_1,a_2$).

\begin{proof}[Proof of \autoref{thm:main-sg-intro}]

Partition the polynomials to two sets. Let $\calL$ be the set of all squares and let $\cQ$
be the subset of irreducible quadratics, thus $\tilde{\cQ} = \cQ \cup \calL$. Denote $|\cQ| = m$, $|\calL| = r$.
Let $\delta =\frac{1}{100}$, 
and denote
\begin{itemize}
\item  $\calP_1 = \{P\in \cQ \mid \text{There are at least } \delta m \text{ polynomials in } \cQ \text { such that } P \text{ satisfies \autoref{thm:structure}\ref{case:span}} \\ \text{\hspace{2.6cm} but not \autoref{thm:structure}\ref{case:rk1} with each of them} \}$.
\item  $\calP_3 = \{P\in \cQ \mid  \text{There are at least } \delta m \text{ polynomials in } \cQ \text { such that } P \text{ satisfies \autoref{thm:structure}\ref{case:2}} \\ \text{\hspace{2.6cm} 
 with each of them} \}$.
\end{itemize}
The proof first deals with the case where $\cQ = \calP_1\cup \calP_3$. We then handle the case that there is $Q \in \cQ \setminus(\calP_1\cup \calP_3)$.

\subsection{The case $\cQ=\calP_1\cup\calP_3$.}

Assume that $\cQ = \calP_1\cup \calP_3$. For our purposes, we may further assume that $\calP_1\cap \calP_3 = \emptyset$, by letting $\calP_1 = \calP_1 \setminus \calP_3$.

\begin{claim}\label{cla:VforP3}
There exists a linear space of linear forms, $V$, such that $\dim(V)=O(1)$ and $\calP_3 \subset \ideal{V}$.
\end{claim}
The intuition  behind the claim is based on the following observation.

\begin{observation}\label{rem:4-2-dim}
If $Q_1,Q_2 \in \cQ$ satisfy \autoref{thm:structure}\ref{case:2} then $\dim(\MS(Q_1)), \dim(\MS(Q_2)) \leq 4$ and $\dim(\MS(Q_1)\cap\MS(Q_2)) \geq 2$. 
\end{observation}

Thus, we have many small dimensional spaces that have large pairwise intersections and we can therefore expect that such a $V$ may exist.

\begin{proof}
We prove the existence of $V$ by explicitly constructing it.
Repeat the following process: Set $V = \{\vec{0}\}$, and $\calP'_3=\emptyset$. At each step consider any $Q \in \calP_3$ such that $Q\notin \ideal{V}$ and set $V =\MS(Q) + V$, and $\calP'_3=\calP'_3 \cup \{Q\}$. Repeat this process as long as possible, i.e, as long as $\calP_3 \not \subseteq \ideal{V}$. We  show next that this process must end after at most $\frac{3}{\delta}$ steps. In particular, $|\calP'_3| \leq \frac{3}{\delta}$. It is clear that at the end of the process it holds that $\calP_3 \subset \ideal{V}$.


\begin{claim}\label{cla:3-case3}
Let $Q\in \cQ$ and $\cB\subseteq \calP'_3$ be the subset of all polynomials in $\calP'_3$ that satisfy \autoref{thm:structure}\ref{case:2} with $\cQ$, then $|\cB| \leq 3$.
\end{claim}

\begin{proof}
Assume towards a contradiction that $|\cB| \geq 4$, and that $Q_1,Q_2,Q_3$ and $Q_4$ are the first $4$ elements of $\cB$ that where added to $\calP'_3$. Denote $U=\MS(Q)$, and  $U_i = U\cap \MS(Q_i)$, for $1\leq i \leq 4$.

As $Q$ satisfies \autoref{thm:structure}\ref{case:2} we have that  $\dim(U) \leq 4$. Furthermore, for every $i$, $\dim(U_i)\geq 2$ (by \autoref{rem:4-2-dim}). As the $Q_i$s were picked by the iterative process, we have that $U_2 \not \subseteq U_1$. Indeed, since $Q_2 \in \ideal{U_2}$, if we had $U_2 \subseteq U_1\subseteq  \MS(Q_1)\subseteq V$, then this would imply that $Q_2\in\ideal{V}$, in contradiction to the fact that $Q_2\in \calP'_3$. Similarly we get that
$U_3 \not \subseteq U_1 + U_2$ and  $U_4 \not \subseteq U_1+U_3 +U_3$. However, as the next simple lemma shows, this is not possible.
\begin{lemma} \label{lem:3-are-V}
Let $V$ be a linear space of dimension $\leq 4$, and let $V_1,V_2,V_3 \subset V$ each of dimension $\geq 2$, such that $V_1\not \subseteq V_2$ and  $V_3\not \subseteq V_2 + V_1$ then $V = V_1+V_2+V_3$.
\end{lemma}
\begin{proof}
As $V_1\not \subseteq V_2$ we have that $\dim(V_1+V_2)\geq 3$. Similarly we get 
$4\leq \dim(V_1+V_2+V_3)\leq \dim(V)=4$.
\end{proof}
Thus, \autoref{lem:3-are-V} implies that $V=U_1+U_2+U_3$ and in particular, $U_4 \subseteq U_1+U_2+U_3$ in contradiction. This completes the proof of \autoref{cla:3-case3}.
\end{proof}

For $Q_i \in \calP'_3$, define $T_i =\{Q\in \cQ \mid Q,Q_i \text{ satisfiy \autoref{thm:structure}\ref{case:2}}\}$. Since $|T_{i}|\geq \delta m$, and as by \autoref{cla:3-case3} each $Q\in \cQ$ belongs to at most $3$ different sets, it follows by double counting that  $|\calP'_3|\leq 3/\delta$.
As in each step we add at most $4$ linearly independent linear forms to $V$, we obtain $\dim(V)\leq \frac{12}{\delta}$.

This completes the proof of \autoref{cla:VforP3}.
\end{proof}

So far $V$ satisfies that $\calP_3 \subset \ideal{V}$. Next, we find a small set of polynomials $\cI$ such that $\cQ \subset \ideal{V}+\spn{\cI}$. 


\begin{claim}\label{cla:Iexists}
There exists a set $\cI\subset \cQ$ such that $\cQ \subset \ideal{V}+\spn{\cI}$ and $|\cI|= O(1/\delta) $. 
\end{claim}

\begin{proof}
As before the proof shows how to construct $\cI$ by an iterative process. 
Set $\cI= \emptyset$ and $\cB=\calP_3$. First add to $\cB$ any polynomial from $\calP_1$ that is in $\ideal{V}$. Observe that at this point we have that $\cB\subset \cQ\cap \ideal{V}$. We now describe another iterative process for the polynomials in $\calP_1$.
In each step pick any $P \in \calP_1\setminus \cB$ such that $P$ satisfies \autoref{thm:structure}\ref{case:span}, but not \autoref{thm:structure}\ref{case:rk1},\footnote{By this we mean that there are many polynomials that together with $P$ span another polynomial in $\cQ$ but not in $\calL$.} with at least $\frac{\delta}{3}m$ polynomials in $\cB$, and add it to both $\cI$ and to $\cB$. Then, we add to $\cB$ all the polynomials $P'\in\calP_1$ that satisfy $P' \in \spn{(\cQ \cap \ideal{V})  \cup \cI}$. Note, that we always maintain that  $\cB \subset \spn{(\cQ \cap \ideal{V})  \cup \cI}$.

We continue this process as long as we can.
Next, we prove that at the end of the process we have that $|\cI| \leq 3/\delta$.

\begin{claim}\label{cla:2nd-proc}
In each step we added to $\cB$ at least $\frac{\delta}{3}m$ new polynomials from $\calP_1$. In particular, $|\cI| \leq 3/\delta$.
\end{claim}
\begin{proof}
Consider what happens when we add some polynomial $P$ to $\cI$. By the description of our process, $P$ satisfies \autoref{thm:structure}\ref{case:span} with at least $\frac{\delta}{3}m$ polynomials in $\cB$. Any $Q\in  \cB$, that satisfies \autoref{thm:structure}\ref{case:span} with $P$,  must span with $P$ a polynomial $P'\in \tilde{\cQ}$. Observe that $P' \notin \calL$ as $Q,P$ do not satisfy \autoref{thm:structure}\ref{case:rk1}, and thus  $P'\in\cQ$. It follows that $P'\in \calP_1$ since otherwise we would have that $P\in \spn{\cB}\subset  \spn{(\cQ\cap \ideal{V}) \cup \cI}$, which implies $P\in\cB$ in contradiction to the way that we defined the process.
Furthermore, for each such $Q\in \cB$ the polynomial $P'$ is unique. Indeed, if there was a $P\neq P'\in\calP_1$ and $Q_1,Q_2\in  \cB$ such that $P'\in \spn{Q_1,P}\cap\spn{Q_2,P}$ then by pairwise independence we would conclude that $P\in \spn{Q_1,Q_2}\subset \spn{\cB}$, which, as we already showed, implies $P\in\cB$ in contradiction. Thus, when we add $P$ to $\cI$ we add at least $\frac{\delta}{3}m$ polynomials  to $\cB$. In particular, the process terminates after at most $3/\delta$ steps and thus $|\cI|\leq 3/\delta$.
\end{proof} 


Consider the polynomials left in $\calP_1\setminus \cB$. As they ''survived'' the process, each of them satisfies the condition in the definition of $\calP_1$  with at most $\frac{\delta}{3}m$ polynomials in $\cB$. From the fact that $\calP_3\subseteq \cB$ and  the uniqueness property we obtained in the proof of \autoref{cla:2nd-proc}, we get that $\calP_1\setminus \cB$ satisfies the conditions of \autoref{def:delta-SGConf} with parameter $\delta/3$ and thus, \autoref{thm:robustSG} implies that $\dim(\calP_1\setminus \cB)\leq O(1/\delta)$. Adding a basis of $\calP_1\setminus \cB$ to $\cI$ we get that $|\cI| = O(1/\delta)$ and every polynomial in $\cQ$ is in $\spn{(\cQ\cap \ideal{V})\cup \cI}$.
\end{proof}

 
We are not done yet as the dimension of $\ideal{V}$, as a vector space, is not a constant.  Nevertheless, we next show how to use \hyperref[thm:SG-linforms]{Sylvester-Gallai theorem} to bound the dimension of $\cQ$ given that  $\cQ \subset \spn{(\cQ\cap \ideal{V})\cup \cI}$. To achieve this we introduce yet another iterative process:  For each $P\in \cQ\setminus \ideal{V}$, if there is quadratic $L$, with $\rank_s(L) \leq 2$, such that  $P + L \in \ideal{V}$, then we set $V = V+\MS(L)$ (this increases the dimension of $V$ by at most $4$). Since this operation increases $\dim\left( \ideal{V}\cap \cQ \right)$ we can remove one polynomial from $\cI$, and thus decrease its size by $1$, and still maintain the property that $\cQ \subset \spn{(\cQ\cap \ideal{V})\cup \cI}$. We repeat this process until either $\cI$ is empty, or none of the polynomials in $\cI$ satisfies the condition of the process. By the upper bound on $|\cI|$ the dimension of $V$ grew by at most  $4|\cI|= O( 1/{\delta})$ and thus it remains of dimension $O( 1/{\delta})=O(1)$. 
At the end of the process we have that $\cQ \subset \spn{(\cQ\cap \ideal{V})\cup \cI}$
and that every polynomial in $P \in \cQ\setminus \ideal{V}$ has $\rank_s(P) > 2$, even if we set all linear forms in $V$ to zero. 

Consider the map $T_{\vaa,V}$ as given in \autoref{def:z-mapping}, for a randomly chosen $\vaa\in[0,1]^{\dim(V)}$. Each polynomial in $\cQ\cap \ideal{V} $ is mapped to a polynomial of the form form $zb$, for some linear form $b$. From \autoref{cla:rank-mod-space}, it follows that every polynomial in $\cQ\setminus \ideal{V}$ still has rank larger than $2$ after the mapping.
Let $$\cA= \{b \mid \text{ some polynomial in } \cQ\cap \ideal{V}  \text{ was mapped to } zb\} \cup T_{\vaa,V}(\calL) \;.$$
We now show that, modulo $z$, $\cA$ satisfies the conditions of \hyperref[thm:SG-linforms]{Sylvester-Gallai theorem}.
Let $b_1,b_2\in\cA$ such that $b_1 \not\in\spn{z}$ and $b_2 \not\in\spn{z,b_1}$.
As $\tilde\cQ$ satisfies the conditions of \autoref{thm:main-sg-intro} we get that there are  polynomials $Q_1,\ldots,Q_4 \in \tilde\cQ$  such that $\prod_{i=1}^{4}T_{\vaa,V}(Q_i) \in \sqrt{\ideal{b_1,b_2}}=\ideal{b_1,b_2}$, where the equality holds as  $\ideal{b_1,b_2}$ is a prime ideal. This fact also implies that, without loss of generality,  $T_{\vaa,V}(Q_4)\in \ideal{b_1,b_2}$. Thus, $T_{\vaa,V}(Q_4)$ has rank at most $2$ and therefore $Q_4\in\calL\cup(\cQ\cap \ideal{V}) $. Hence, $T_{\vaa,V}(Q_4)$ was mapped to $zb_4$ or to $b_4^2$. In particular, $b_4\in\cA$. \autoref{cla:res-z-ampping} and \autoref{cla:still-indep} imply that $b_4$ is neither a multiple of $b_1$ nor a multiple of $b_2$, so it must hold that $b_4$ depends non-trivially on both $b_1$ and $b_2$. Thus, $\cA$ satisfies the conditions of \hyperref[thm:SG-linforms]{Sylvester-Gallai theorem} modulo $z$. It follows that $\dim(\cA)=O(1)$.

The argument above shows that the dimension of $T_{\vaa,V}(\calL\cup(\cQ\cap \ideal{V}) ) =O(1)$. \autoref{cla:z-map-dimension} implies that if we denote $U = \spn{\calL \cup \MS(\cQ\cap \ideal{V})}$ then  $\dim(U )$ is $O(1)$. As $\cQ \subseteq \spn{(\cQ\cap \ideal{V}) \cup \cI}$, we obtain that $\dim(\tilde\cQ)=\dim(\calL\cup \cQ) = O(1)$, as we wanted to show. 
%
%

This completes the proof of \autoref{thm:main-sg-intro} for the case $\cQ=\calP_1\cup\calP_3$.

\subsection{The case $\cQ\neq\calP_1\cup\calP_3$.}

In this case there is some polynomial $Q_o \in \cQ \setminus ({\calP}_1\cup {\calP}_3)$. In particular,  $Q_0$ satisfies \autoref{thm:structure}\ref{case:rk1} with at least $(1-2\delta)m$ of the polynomials in $\cQ$; of the remaining polynomials, at most $\delta m$ satisfy \autoref{thm:structure}\ref{case:span} with $Q_o$; 
and, $Q_o$ satisfies \autoref{thm:structure}\ref{case:2} with at most $\delta m$ polynomials.
 Let 
\begin{itemize}
\item $\cQ_1 =  \{P \in \cQ \mid P,Q_o \text{ satisfiy \autoref{thm:structure}\ref{case:rk1} }\} \cup \{Q_o\}$
\item $\cQ_2 = \{P \in \cQ \mid P,Q_o \text{ do not satisfiy \autoref{thm:structure}\ref{case:rk1} }\}$ 
\item $m_1 = |\cQ_1|$, $m_2 = |\cQ_2|$.
\end{itemize}
As $Q_o \notin \calP_1 \cup \calP_3$ we have that $m_2 \leq 2\delta m$ and $m_1 \geq (1-2\delta)m$.
These properties of $Q_o$ and $\cQ$ are captured by the following definition.

\begin{definition}\label{def:Q-dom}
Let  $\mathcal{Q}_1 = \{Q_o,Q_1,\ldots,Q_{m_1}\}$ and $\mathcal{Q}_2 = \{\MVar{P}{m_2}\}$ be sets of irreducible homogeneous quadratic polynomials. Let $\calL = \{\MVar{\ell^2}{r}\}$ be a set of squares of homogeneous linear forms. We say that $\tilde{\cQ} = \cQ \cup \calL$ where $\cQ=\cQ_1\cup\cQ_2$ is a {$(Q_o,m_1,m_2)$-set} if it satisfies the following:
\begin{enumerate}
\item $\tilde{\cQ}$ satisfy the conditions in the statement of \autoref{thm:main-sg-intro}.
\item $m_1 > 5m_2+2$.
\item For every $j\in [m_1]$, there are  linear forms $a_j,b_j$ such that  $Q_j = Q_o + a_j b_j$.
\item For every $i\in [m_2]$, every non-trivial linear combination of $P_i$ and $Q_o$ has rank at least $2$.\label{item:Qdom>2}
\item At most $m_2$ of the polynomials in $\cQ$  satisfy \autoref{thm:structure}\ref{case:2} with $Q_o$. \label{item:Qdom-case3}
\end{enumerate}
\end{definition}

By the discussion above, the following theorem is what we need in order to complete the proof for the case   $\cQ\neq\calP_1\cup\calP_3$.

\begin{theorem}\label{thm:Q-dom-gen}
Let $\tilde{\cQ}$ satisfy the conditions of \autoref{def:Q-dom}, then $\dim{\tilde{\cQ}} = O(1)$.
\end{theorem}

We prove this theorem in \autoref{sec:q-dom}. This concludes the proof of \autoref{thm:res}. 
\end{proof}

\section{Proof of Theorem~\ref{thm:Q-dom-gen}}\label{sec:q-dom}

In this section we prove \autoref{thm:Q-dom-gen}. The proof is divided to two parts according to whether the polynomial $Q_o$ in \autoref{def:Q-dom} is of high rank (\autoref{cla:strong-Q-dom-hr-main}) or of low rank (\autoref{cla:strong-Q-dom-lr-main}). Each part is also divided to two -- first we consider what happens when $m_2=0$ and then the general case where $m_2\neq 0$. The reason for this split is that when $Q_o$ is of high rank then we know, e.g., that it cannot satisfy \autoref{thm:structure}\ref{case:2} with any other polynomial. Similarly any polynomial satisfying \autoref{thm:structure}\ref{case:rk1} with $Q_o$ is also of high rank and cannot satisfy \autoref{thm:structure}\ref{case:2} with any other polynomial. The reason why we further break the argument to weather $m_2=0$ or not, is that when $m_2= 0$ all the polynomials  are of the form $Q_o+ab$ for some linear forms $a,b$, which means we have fewer cases to analyse.
While this seems a bit restrictive, the general case is not much harder and most of the ideas there already appear in the case $m_2=0$. 

Throughout the proof we use the notation of \autoref{def:Q-dom}. In particular, each $Q_i\in\cQ_1$ is of the form $Q_i=Q_o+a_ib_i$.


\subsection{$Q_o$ is of high rank}

In this subsection we assume that $\tilde{\cQ}$ is a \hyperref[def:Q-dom]{$(Q_o,m_1,m_2)$-set} for some quadratic $Q_o$ of rank at least $\rkq$, this constant is arbitrary, as we just need it to be large enough.
The following observation says that for our set $\cQ$ we will never have to consider 
\autoref{thm:structure}\ref{case:2}. 

\begin{observation}\label{obs:case3}
For $\tilde{\cQ} = \cQ \cup \calL$ that satisfy \autoref{def:Q-dom} with $\rank_s(Q_o)\geq \rkq$, for every $j \in [m_1]$ the rank of $Q_j$ is at least $\rkq-1>2$ and so $Q_j$ never satisfies \autoref{thm:structure}\ref{case:2} with any other polynomial in $\tilde{\cQ}$.

\end{observation}

Our goal in this subsection is to prove the next claim.

\begin{claim}\label{cla:strong-Q-dom-hr-main}
Let $\tilde{\cQ} = \cQ \cup \calL$ be a \hyperref[def:Q-dom]{$(Q_o,m_1,m_2)$-set} with $\rank_s(Q_o)\geq \rkq$.
Then $\dim(\spn{\tilde{\cQ}}) = O(1)$.
\end{claim}

We break the proof of \autoref{cla:strong-Q-dom-hr-main} to two steps. First we handle the case $m_2=0$ and then the case $m_2\neq 0$.


\subsubsection{The case $m_2=0$}\label{sec:m2=0-hr}

In this subsection we prove the following version of \autoref{cla:strong-Q-dom-hr-main} for the case $m_2=0$.

\begin{claim}\label{cla:strong-Q-dom-hr}
Let $\tilde{\cQ} = \cQ \cup \calL$ be a \hyperref[def:Q-dom]{$(Q_o,m_1,0)$-set} with $\rank_s(Q_o)\geq \rkq$.
Then, for $a_i,b_i, \ell_j$ as in \autoref{def:Q-dom}, $\dim(\spn{\MVar{a}{m_1}, \MVar{b}{m_1}, \MVar{\ell}{r}})\leq 7$. In particular, $\dim(\spn{\cQ}) \leq 8$.
\end{claim}

We first show some properties satisfied by the products $\lbrace a_1b_1, \ldots, a_{m_1}b_{m_1}\rbrace$.

\begin{remark}\label{rem:alpha}
	For $\ell_i^2 \in \calL$ we can write $\ell^2_i = 0\cdot Q_o + \ell_i\ell_i$. Thus, from now on we can assume that every $Q_i \in \tilde{\cQ}$ is of the form $Q_i = \alpha_i Q_o + a_ib_i$, for $\alpha_i\in\{0,1\}$, and when $\alpha_i =0$ it holds that $a_i=b_i$. We shall use the convention that for $i\in\{m_1+1,\ldots,m_1+r\}$, $a_i = \ell_{i-m_1}$.
\end{remark}

\begin{claim}\label{cla:span-rank1-hr}
Let $\tilde{\cQ} = \cQ \cup \calL$ be a \hyperref[def:Q-dom]{$(Q_o,m_1,0)$-set} with $\rank_s(Q_o)\geq \rkq$, and let  $Q_i=Q_o+a_ib_i$ and $Q_j= Q_o+a_jb_j$ be polynomials in $\cQ = \cQ_1$. 
\begin{enumerate}
\item If $Q_i$ and $Q_j$ satisfy \autoref{thm:structure}\ref{case:span} then there exists $k\in [m_1+r]$ such that for some $\alpha,\beta \in \C\setminus\{0\}$
\begin{equation} \label{eq:1-alpha-hr}
 \alpha a_ib_i + \beta a_jb_j = a_k b_k. 
\end{equation}
 \item If $Q_i$ and $Q_j$ satisfy \autoref{thm:structure}\ref{case:rk1} then there exist two linear forms, $c$ and $d$ such that
\begin{equation}
 a_ib_i -  a_jb_j = cd.
\end{equation}
\end{enumerate}
\end{claim}

The claim only considers \autoref{thm:structure}\ref{case:span} and \autoref{thm:structure}\ref{case:rk1} as by \autoref{obs:case3} we know that $Q_i,Q_j$ do not satisfy \autoref{thm:structure}\ref{case:2}.
 Note that the guarantee of this claim is not sufficient to conclude that the dimension of $\MVar{a}{m_1}, \MVar{b}{m_1}$ is bounded. The reason is that $ c$ and $d $ are not necessarily part of the set. For example if for every $i$, $a_ib_i = x_i^2 - x_1^2$. Then every pair, $Q_i, Q_j$ satisfy \autoref{thm:structure}\ref{case:rk1}, but the dimension of $\MVar{a}{m_1}, \MVar{b}{m_1}$ is unbounded.

\begin{proof}[Proof of \autoref{cla:span-rank1-hr}]
If $Q_i,Q_j$ satisfy \autoref{thm:structure}\ref{case:span} then there are constants  $\alpha, \beta \in \C$ and $k \in [m_1+r]\setminus \{i,j\}$ such that $\alpha (Q_o+a_ib_i) + \beta ( Q_o+a_jb_j) = \alpha Q_i+\beta Q_j = Q_k = \alpha_k Q_o+a_kb_k$. Rearranging we get that \[\alpha a_ib_i + \beta a_jb_j -a_kb_k = (\alpha_k-(\alpha+\beta))Q_o\;.\] From the fact that $\rank_s(Q_o) \geq \rkq$, it must be that $\alpha_k-(\alpha+\beta) = 0$. Hence,
\begin{equation} \label{eq:span-ak-hr}
 \alpha a_ib_i + \beta a_jb_j = a_kb_k
\end{equation}
and \eqref{eq:1-alpha-hr} holds. Observe that $\alpha, \beta \neq 0$ as otherwise we will have two linearly dependent polynomials in $\cQ$.

If $Q_i,Q_j$ satisfy \autoref{thm:structure}\ref{case:rk1} then there are $\alpha, \beta \in \C$ and two linear forms $c$ and $d$ such that $ \alpha (Q_o+a_ib_i) + \beta ( Q_o+a_jb_j) = cd$, and again, by the same argument, we get that  $\beta = -\alpha$, and that, without loss of generality,
\begin{equation*}
  a_ib_i -a_jb_j = cd. \qedhere
\end{equation*}
\end{proof}

Let $V_i \eqdef \spn{a_i,b_i}$. We next show that the different spaces $V_i$ satisfy some non-trivial intersection properties.

\begin{claim}\label{cla:V_i intersection-hr}
Let $\tilde{\cQ}$ be a \hyperref[def:Q-dom]{$(Q_o,m_1,0)$-set} such that $\rank_s(Q_o)\geq \rkq$. If for some $i\in [m_1]$ we have $\dim(V_i)=2$ then for every $j \in [m_1]$ it holds that  $\dim(V_j\cap V_i) \geq 1$. In particular it follows that if $dim(V_j)=1$ then $V_j\varsubsetneq V_i$.
\end{claim}
\begin{proof}
This follows immediately from \autoref{cla:span-rank1-hr} and \autoref{cla:intersection}. 
\end{proof}


Next we use this fact to conclude some structure on the set of pairs $(a_i,b_i)$.
 
\begin{claim}\label{cla:z-exist}
Let $\tilde{\cQ}$ be as in \autoref{cla:strong-Q-dom-hr}. If $\dim(\spn{a_i,b_i})>3$ then there is a linear space of linear forms, $V$ such that $\dim(V) \leq 4$, and for all $i\in [m_1+r]$, $b_i \in \spn{a_i,V}$ or $a_i \in \spn{b_i,V}$. 
\end{claim}

\begin{proof}
Consider the set of all $V_i$'s of dimension $2$. Combining   \autoref{cla:span-rank1-hr} and \autoref{cla:linear-spaces-intersaction } we get that either  $\dim(\bigcup_{i=1}^m V_i) \leq 3$ or $\dim(\bigcap_{i=1}^m V_i) =1$. If $\dim(\bigcup_{i=1}^m V_i) \leq 3$ then $V = \bigcup_{i=1}^m V_i$ is the linear space promised in the claim. If $\bigcap_{i=1}^m V_i) =1$
 there is a linear form, $w$, such that $\spn{w} = \dim(\bigcap_{i=1}^m V_i)$. It follows that for every $ i \in [m_1]$ there are constants $\epsilon_i,\delta_i$ such that, with out loss of generality, $ b_i = \epsilon_i a_i + \delta_i w$. Note that if $\dim(V_i) = 1$ this representation also holds with $\delta_i = 0$, and thus $V = \spn{w}$. is the linear space promised in the claim.
\end{proof}

From now on we assume there is a linear space of linear forms, $V$ such that $\dim(V) \leq 4$ and for every $i \in [m_1+r]$ it holds that $b_i = \epsilon_i a_i + v_i$ (we can do this by replacing the roles of $a_i$ and $b_i$ if needed). Indeed, if $\dim(\spn{a_i,b_i})>3$ then this follows from \autoref{cla:z-exist} and otherwise we can take $V=\spn{a_i,b_i}$. Thus, following \autoref{rem:alpha},  every polynomial in $\cQ$ is of the form $\alpha_i Q + a_i(\epsilon_i a_i +v_i)$ and for polynomials in  $\calL$ we have that $\alpha_i = 0$, $\epsilon_i = 1$ and $ v_i = 0$. 

The following claim is the crux of the proof of \autoref{cla:strong-Q-dom-hr}. It shows that, modulo $V$,  the set $\{\MVar{a}{m_1+r}\}$ satisfies the Sylvester-Gallai theorem..

\begin{claim}\label{cla:stisfy-sg-lines-hr}
Let $i\neq j \in [m_1+r]$ be such that $a_i \notin V$ and $a_j \notin \spn{a_i,V}$. Then, there is $k \in [m_1+r]$ such that $a_k \in \spn{a_i,a_j,V}$ and $a_k \notin \spn{a_i,V}\cup \spn{a_j,V}$.
\end{claim}

\begin{proof}

We split the proof to three cases (recall \autoref{rem:alpha}): Either 
\begin{enumerate*}[label=(\roman*)]
  \item $\alpha_i = \alpha_j = 1$, or 
  \item $\alpha_i = 1, \alpha_j = 0$ (without loss of generality), or
  \item $\alpha_i=\alpha_j = 0$.
\end{enumerate*} 
Recall that $\alpha_i=0$ if and only if $i\in\{m+1,\ldots,m+r\}$. 

\begin{enumerate}[label=(\roman*)]
	\item  $\alpha_i= \alpha_j = 1$. 	
\autoref{cla:span-rank1-hr} implies that there are two linear forms $c$ and $d$ such that $cd$ is a nontrivial linear combination of ${a_j(\epsilon_j a_j + v_j), a_i(\epsilon_i a_i + v_i)}$. We next show that  without loss of generality $c$ depends non-trivially on both $a_i$ and $a_j$.

\begin{lemma}\label{lem:span-c-hr}
In the current settings, without lost of generality, $c= \mu a_i + \eta a_j$ where $\mu,\eta \neq 0$.
\end{lemma}
\begin{proof}
Setting $a_i=0$ gives that, without loss of generality, $cd \equiv_{a_i} a_j(\epsilon_j a_j + v_j)$ and as $a_j\not\in \spn{a_i,V}$ we have that $cd\not\equiv_{a_i}0$. Thus,  without loss of generality $c \equiv_{a_i} \eta a_j$, for some non-zero $\eta$. Let $\mu$ and $\eta$ be such that  $c = \mu a_i + \eta a_j$. We will now show that $\mu  \neq 0$. Indeed, if this was not the case then we would have that $cd = \eta a_j d$.
This means that $ a_i(\epsilon_i a_i + v_i) \in \spn{a_j(\epsilon_j a_j + v_j), \eta a_j d}$ (since the linear dependence was non-trivial) setting $a_j =0$ we see that either $a_i$, or $ \epsilon_i a_i + v_i$ in $\spn{a_j}$, which contradicts our assumption. 
\end{proof}

\autoref{eq:1-alpha-hr} and \autoref{lem:span-c-hr} show that if $Q_i$ and $Q_j$ satisfy \autoref{thm:structure}\ref{case:span}, i.e. they span $Q_k$ (for $k\not\in\{i,j\}$), then one of $a_k,\epsilon_k a_k + v_k$ is a non-trivial linear combination of $a_i$ and $a_j$. Thus, modulo $V$, $a_k$ is in the span of $a_i$ and $a_j$, which is what we wanted to show. 
	
	We next handle the case where $Q_i$ and $Q_j$ satisfy \autoref{thm:structure}\ref{case:rk1}. Let $cd$ be a product of linear forms in the span of $Q_i$ and $Q_j$. From \autoref{lem:span-c-hr} we can assume that $c=\mu a_i+\eta a_j$ with $\mu\eta\neq 0$. In particular, this means that $\sqrt{\ideal{Q_i,Q_j}} = \sqrt{\ideal{cd,Q_j}}$. 
	
	The assumption that $\rank_s(Q_o)\geq \rkq$ implies that $Q_j$ is irreducible even after setting $c=0$. It follows that if a product of irreducible polynomials satisfy $\prod_i A_i \in  \sqrt{\ideal{cd,Q_j}}$ then, after setting $c=0$, some $A_i$ is divisible by ${Q_j}|_{c=0}$. Thus, there is a multiplicand that is equal to $ \alpha Q_j + ce$ for some linear form $e$. In particular, there must be a polynomial $Q_k$, $k\in[m_1+r]\setminus \{i,j\}$, such that  $Q_k = \alpha Q_j + ce$.
	If $\alpha = 0$ then it holds that $Q_k = a_k^2 = ce$ and therefore $a_k$ satisfies the claim. Otherwise, as before, the rank condition on $Q_o$ implies that $\alpha = 1$ and thus $ a_k(\epsilon_k a_k +v_k) =  a_j(\epsilon_j a_j + v_j) + (\mu a_i + \eta a_j)e$. 
	Consider what happens when we set $a_j=0$. We get that $a_k(\epsilon_k a_k +v_k) \equiv_{a_j} \mu a_i e$.  Note that it cannot be the case that $e\equiv_{a_j}0$ as this would imply that $a_k\in\spn{a_j,v_k}$ and in turn, this implies that $a_i \in \spn{a_j,V}$ in contradiction to the choice of $a_i$ and $a_j$. Thus, we get that either $a_k$ or $\epsilon_k a_k +v_k$ are equivalent to $a_i$ modulo $a_j$. We next show that if either of them depends only on $a_i$, then we get a contradiction. Thus, we are left in the case that $a_k = \lambda a_i$ (the case $\epsilon_k a_k + v_k = \lambda a_i$ is equivalent).  Since $Q_k = Q_o+  \lambda a_i\left(\epsilon_k  \lambda a_i + v_k\right) = Q_j + ce$ and we have that $Q_i = Q_o + a_i(\epsilon_i a_i + v_i) = Q_j + cd$ we get by subtracting $Q_i$ from $Q_k$ that 
	\[a_i\left((\lambda^2\epsilon_k - \epsilon_i)a_i + (\lambda v_k -v_i)\right) =\lambda a_i(\epsilon_k  \lambda a_i + v_k) - a_i(\epsilon_i   a_i + v_i) = Q_k-Q_i=c(e-d) \;,\]
	and clearly neither side of the equation is zero since $Q_i\neq Q_k$. 
	This implies that $c \in \spn{a_i,V}$, in contradiction. Thus, in this case too we get that $a_k$ satisfies the claim.
	
	\item $\alpha_i= 1, \alpha_j = 0$. In this case, $Q_i,Q_j$ must satisfy \autoref{thm:structure}\ref{case:rk1}, as $0 \cdot Q_i + Q_j = a_j^2$. As before,  the  assumption that $\rank_s(Q_o)\geq \rkq$ implies that $Q_i$ is irreducible even after setting $a_j=0$. It follows that if a product of irreducible polynomials satisfy $\prod_t A_t \in  \sqrt{\ideal{a_j^2,Q_i}}$ then, after setting $a_j=0$, some $A_t$ is divisible by ${Q_i}|_{a_j=0}$. In our case we get that there is a multiplicand that is equal to $ \alpha Q_i + a_je$ for some linear form $e$. In particular, there must be a polynomial $Q_k$, for $k\in[m_1+r]\setminus \{i,j\}$, such that  $Q_k = \alpha Q_i + a_je$.
	If $\alpha = 0$ it follows that $Q_k$ is reducible and thus of the form $Q_k = a_k^2 = a_je$ which is a contradiction to pairwise linear independence (as $Q_k \sim Q_j$).  Thus $\alpha = \alpha_k = 1$, and $a_k(\epsilon_ka_k + v_k) = a_i(\epsilon_i a_i + v_k) + a_je$. As before, we can conclude that $a_k \in \spn{a_i,a_j,V}$ and that it cannot be the case that $a_k \in \spn{a_i,V} \cup \spn{a_j,V}$ (as by rearranging the equation we will get a contradiction to the fact that $a_j \notin \spn{a_i,V}$),  which is what we wanted to show.

	
	\item  $\alpha_i= \alpha_j = 0$. Then $\sqrt{\ideal{Q_i,Q_j}} = \ideal{a_i,a_j}$ is a prime ideal. It follows that there is $k \in [m_1+r] \setminus \{i,j\}$ such that $Q_k \in \ideal{a_i,a_j}$ the  rank condition on $Q_o$ implies that $\alpha_k = 0$ and therefore $a_k$ is a non-trivial linear combination of $a_i$ and $a_j$, which is what we wanted to show.

\end{enumerate}


This completes the proof of \autoref{cla:stisfy-sg-lines-hr}.
\end{proof}

We can now prove  \autoref{cla:strong-Q-dom-hr}.

\begin{proof}[Proof of  \autoref{cla:strong-Q-dom-hr}]
%
%


\autoref{cla:stisfy-sg-lines-hr} implies that any two linear functions in $\{\MVar{a}{m_1+r}\}$ that are linearly independent modulo $V$, span (modulo $V$) a third function in the set. This implies that if we project all the linear functions to the perpendicular space to $V$ then they satisfy the usual condition of the \hyperref[thm:SG-linforms]{Sylvester-Gallai theorem} and thus the dimension of the projection is at most $3$. As $\spn{ \MVar{a}{m_1}, \MVar{b}{m_1},a_{m_1+1},\ldots,a_{m_1+r}} \subseteq \spn{\MVar{a}{m_1+r},V}$, we get  that  $\dim(\{\MVar{a}{m_1}, \MVar{b}{m_1},a_{m_1+1},\ldots,a_{m_1+r}\}) \leq 3+\dim(V) \leq 7$, as claimed.
\end{proof}

Thus far we have proved  \autoref{cla:strong-Q-dom-hr} which is a restriction of  \autoref{cla:strong-Q-dom-hr-main} to the case $m_2=0$. In the next subsection we handle the general case $m_2\neq 0$.

\subsubsection{The case $m_2\neq 0$.}

In this subsection we prove  \autoref{cla:strong-Q-dom-hr-main}.
We shall assume  without loss of generality that $m_2\neq 0$.
We first show that each $P_i \in \cQ_2$ (recall \autoref{def:Q-dom}) is either a rank-$2$ quadratic, or it is equal to $Q_o$ plus a rank-$2$ quadratic.

\begin{claim}\label{cla:no-rank-3}
Let $\tilde{\cQ}$ be a \hyperref[def:Q-dom]{$(Q_o,m_1,m_2)$-set} such that $\rank_s(Q_o)\geq \rkq$. Then for every $i \in [m_2]$ there exists $\gamma_i\in \C$  such that $\rank_s(P_i-\gamma_i Q_o)) = 2$.
\end{claim}

\begin{proof}
Fix $i\in [m_2]$. We shall analyse, for each $j \in [m_1]$, which case of \autoref{thm:structure} $Q_j$ and $P_i$ satisfy.
From \autoref{obs:case3} we know that $P_i$ does not satisfy \autoref{thm:structure}\ref{case:2} with any $Q_j$.
We start by analysing what happens when $P_i$ and  $Q_j$ satisfy \autoref{thm:structure}\ref{case:rk1}.  By definition, there exist linear forms  $a', b'$ and non zero constants $\alpha, \beta \in \C$, such that $\alpha P_i + \beta Q_j = a'b'$  and thus, 
\begin{equation}
\label{eq:case2}
P_i = \frac{1}{\alpha}\left(a'b' - \beta \left(Q_o+a_j b_j\right)\right) = \frac{-\beta}{\alpha}Q_o + \left(\frac{1}{\alpha} a'b' - \frac{\beta}{\alpha} a_j b_j\right) \;.
\end{equation} 
Hence, the statement holds with $\gamma_i = - \frac{\beta}{\alpha}$. Indeed, observe that the $\rank_s$ of $(\frac{1}{\alpha} a'b' - \frac{\beta}{\alpha} a_j b_j)$ cannot be $1$ as this will contradict \autoref{item:Qdom>2}  in \autoref{def:Q-dom}.

Thus, the only case left to consider is when $P_i$ satisfies \autoref{thm:structure}\ref{case:span} alone with all the $Q_j$'s. If for some $j\in [m_1]$ there is $j'\in [m_1]$ such that $Q_{j'}\in \spn{Q_j,P_i}$, then there are $\alpha, \beta \in \C\setminus\{0\}$, for which $P_i =\alpha Q_j + \beta Q_{j'}$ and then 
\[
P_i = (\alpha + \beta)Q_o + \alpha a_j b_j + \beta a_{j'} b_{j'},
\] and the statement holds with $\gamma_i =  \beta+ \alpha$. 
So, let us assume that for every $j \in [m_1]$, there is $t_j \in [m_2]$ such that $P_{t_j} \in \spn{Q_j,P_i}$. As $5m_2+2<m_1$ there must be $j'\neq j'' \in [m_1]$ and $t'\in [m_2]$ such that $P_{t'} \in \spn{Q_{j'}, P_i}$ and $ P_{t'} \in \spn{Q_{j''}, P_i}$. Since $\mathcal{Q}$ is a set of pairwise linearly independent polynomials, we can deduce that $\spn{P_i, P_{t'}} = \spn{Q_{j'}, Q_{j''}}$. In particular there exist  $\alpha, \beta \in \C$, for which $P_i =\alpha Q_j + \beta Q_{j'}$, which, as we already showed, implies what we wanted to prove.
\end{proof}

For simplicity, rescale $P_i$ so that $P_i = \gamma_i Q_o+L_i$ with $\rank_s(L_i)=2$ and $\gamma_i \in \{0,1\}$. Clearly $\cQ$ still satisfies the conditions of \autoref{def:Q-dom} after this rescaling, as it does not affect the vanishing conditions or linear independence. The next claim shows that even in the case $m_2\neq 0$, the linear forms $\lbrace \MVar{a}{m_1}, \MVar{b}{m_1}\rbrace$ ``mostly'' belong to a low dimensional space (similar to \autoref{cla:strong-Q-dom-hr}). 

\begin{claim}\label{cla:small-space-span}
Let $\tilde{\cQ}$ be a \hyperref[def:Q-dom]{$(Q_o,m_1,m_2)$-set} such that $\rank_s(Q_o)\geq \rkq$. 
Then, there exists a subspace $V$ of linear  forms such that $\dim(V) \leq 4$ and that for at least $m_1-m_2 $ indices $j\in [m_1]$ it holds that $a_j,b_j \in V $. Furthermore, there is a polynomial $P\in\cQ_2$ such that $P = \gamma Q_o + L$ and $\MS(L) = V$.
\end{claim}
\begin{proof}
Let $P_1 = \gamma_1Q_o+L_1$ where $\rank_s(L_1) =2$. To simplify notation we drop the index $1$ and only talk of $P$, $L$ and $\gamma$. Set $V = \MS(L)$.
As before, \autoref{obs:case3} implies that $P$ cannot satisfy \autoref{thm:structure}\ref{case:2} with any $Q_j\in \cQ_1$.

Let $Q_j\in\cQ_1\cup\calL$. If $Q_j, P$ satisfy \autoref{thm:structure}\ref{case:2}, then $\alpha_j = 0$ and $Q_j  =a_j^2$. By the rank condition on $Q_o$ it follows that $\gamma = 0$  and therefore $a_j \in \MS(L) = V$.

Let $Q_j\in \cQ_1\cup\calL$ be such that $Q_j$ and $P$ satisfy \autoref{thm:structure}\ref{case:rk1}. This means that there are two linear forms $e,f$, and non zero $\alpha, \beta \in \C$ for which $\alpha P - \beta Q_j = ef$, and so,
\begin{equation}
(\alpha \gamma -\beta\alpha_j)Q_o = -\alpha L + \beta a_jb_j +ef
\end{equation}

As we assumed that $\rank_s(Q_o) \geq \rkq$ this implies that $\alpha \gamma-\beta\alpha_j = 0$ and thus $\beta a_jb_j + ef = \beta L$.  \autoref{cla:irr-quad-span} implies that $e, f, a_{j}, b_{j} \in V$. 

We have shown that $V$ contains all $a_j,b_j$ that come from polynomials satisfying  \autoref{thm:structure}\ref{case:rk1} with $P$.

Let $j\in[m_1]$ be such that $P$ and $Q_j$ satisfy \autoref{thm:structure}\ref{case:span} but not \autoref{thm:structure}\ref{case:rk1}, i.e, they span another polynomial in $\tilde{\cQ}\setminus\calL$. If this polynomial is in $\cQ_1$, i.e. there exists $j' \in [m_1]$ such that $Q_{j'}\in \spn{P,Q_j}$ then $P = \alpha Q_j + \beta Q_{j'}$ and as before we would get that $a_{j'} ,b_{j'}, a_{j}, b_{j} \in V$.

All that is left is to bound the number of $j\in[m_1]$ so that $P$ and $Q_j$ span a polynomial in $\cQ_2$. If there are more than $m_2$ such indices $j$ then, by the pigeonhole principle, for two of them, say $j,j'$ it must be the case that there is some $i \in [m_2]$ such that $P_{i}\in \spn{P,Q_j}$ and $P_{i}\in \spn{P,Q_{j'}}$. As our polynomials are pairwise independent this implies that $P \in \spn{Q_j,Q_{j'}}$, and as before we get that $a_{j'},b_{j'}, a_{j},b_{j}\in V$. 

It follows that the only $j$'s for which we may have $a_j,b_j\not\in V$ must be such that $Q_j$ and $P$ span a polynomial in $\cQ_2$, and no other $Q_{j'}$ spans this polynomial with $P$. Therefore, there are at most $m_2$ such ``bad'' $j$'s and the claim follows. 
\end{proof}

\begin{remark}\label{rem:QnotV}
The proof of  \autoref{cla:small-space-span} implies that if $Q_i=\alpha_iQ_o+a_ib_i \in \cQ_1$ satisfies that $\{a_i,b_i\}\not\subseteq V$ then
it must be the case that $Q_i$ and $P$ span a polynomial $P_j \in \cQ_2$. 
\end{remark}

\begin{claim}\label{cla:V-map-Q-dom}
Let $\tilde{\cQ}$ be a \hyperref[def:Q-dom]{$(Q_o,m_1,m_2)$-set} such that $\rank_s(Q_o)\geq \rkq$. 
Then there exists a $4$-dimensional linear space $V$, 
such that for every $P_i\in \tilde{\cQ}$ either $P_i$ is defined over $V$, or there is a quadratic polynomial $P'_i$ and a linear form $v_i$ that are defined over $V$, and  a linear form $c_i$, such that $P_i = Q_o +P'_i + c_i(\epsilon_i c_i + v_i )$, or $P_i = c_i^2$. 
\end{claim}
\begin{proof}
 \autoref{cla:small-space-span} implies the existence of a polynomial $P = \gamma Q_o +  L \in \cQ_2$ and $4$-dimensional linear space $V=\MS(L)$ such that the set
$\cI = \lbrace Q_j\mid j\in [m_1]\text{ and }a_j,b_j \in V\rbrace$ satisfies  $\vert \cI \vert \geq m_1-m_2$. We will prove that $V$ is the space guaranteed in the claim. We first note that every $P_i \in \cI$ satisfies the claim with $P'_i = a_ib_i$ and $v_i=c_i =0$, and clearly for
$Q_i \in \calL$ the claim trivially holds.

Consider $Q_i \in \cQ_1 \setminus \cI$. By \autoref{rem:QnotV} it must be the case that $Q_i$ and $P$ span a polynomial $P_j \in \cQ_2$. Hence, there are $\alpha,\beta \in \C\setminus \{0\}$ such that $P_j = \alpha P + \beta Q_i$. From \autoref{cla:no-rank-3} we get that $P_j = \gamma_j Q_o + L_j$ and thus
 \[(\gamma_j -\alpha \gamma-\beta)Q_o = \alpha L + \beta a_ib_i - L_j \;.\]
As $\rank_s(Q_o) \geq \rkq$ it follows that $(\gamma_j -\alpha \gamma-\beta) = 0$ and $\alpha L + \beta a_ib_i = L_j$.  \autoref{cla:ind-rank} implies that $\spn{a_i,b_i} \cap V \neq \{\vec{0}\}$ and therefore there is $v_i\in V$ such that, without loss of generality, $b_i = \epsilon_i a_i + v_i$, for some constant $\epsilon_i$. Thus, the claimed statement holds for $Q_i$ with $c_i = a_i$ and $Q'_i = 0$. I.e., $Q_i = Q_o + 0+ a_i(\epsilon_i a_i + v_i )$.
 
Consider a polynomial $P_i=\gamma_iQ_o+L_i\in\cQ_2$.

If $\gamma_i=0$ then by rank argument we see that $P_i$ cannot satisfy  \autoref{thm:structure}\ref{case:rk1} nor  \autoref{thm:structure}\ref{case:2} with any polynomial in $\cQ_1$. Hence it must satisfy \autoref{thm:structure}\ref{case:span} with all the polynomials in $\cQ_1$. Therefore, by the pigeonhole principle $P_i$ must be spanned by two polynomials in $\cI$. Note that in this case we get that $P_i = L_i$ is a polynomial defined over $V$.

Assume then that $\gamma_i=1$. If  $P_i$ is spanned by $Q_j$ and $Q_{j'}$ such that  $j,j'\in \cI$, then, as before, $\MS(L_i) \subseteq \spn{a_jb_j , a_{j'}b_{j'}}$ and hence $L_i$ is a function of the linear forms in $V$. Thus, the statement holds with $P'_i = L$ and $v_i=c_i =0$.

The only case left to consider is when $\gamma_i =1$ and  every polynomial $Q_j$, for $j\in\cI$, that satisfies \autoref{thm:structure}\ref{case:span} with $P_i$, does not span with $P_i$ any polynomial in $\{Q_j \mid j\in  \cI\} \cup \calL$. Note that in such a case it must hold that $Q_j$ spans with $P_i$ a polynomial in $\{Q_j \mid j\in [m_1]\setminus  \cI\} \cup \cQ_2$. Observe that since our polynomials are pairwise linearly independent, if two polynomials from $\cI$ span the same polynomial with $P_i$ then $P_i$ is in their span and we are done. From 
$$\left|\{Q_j \mid j\in [m_1]\setminus  \cI\} \cup \cQ_2\right| \leq \left(m_1 - |\cI| \right) + m_2 \leq 2m_2 < m_1 -m_2 -2 \leq |\cI| -2 \;,$$
we see that for $P_i$ to fail to satisfy the claim it must be the case that 
it satisfies \autoref{thm:structure}\ref{case:rk1} with at least $2$ polynomials whose indices are in $\cI$. 
 Let $Q_j,Q_{j'} \in \cI$ be two such polynomials. In particular,  there are four linear forms $c,d,e$ and $f$ and scalars $\epsilon_j,\epsilon_{j'}$, such that
\begin{equation}\label{eq:P_i-twice-2}
P_i - \varepsilon_j Q_j = cd \quad \text{and} \quad P_i - \varepsilon_{j'} Q_{j'} = ef \;.
\end{equation} 
Equivalently, 
\begin{equation*}
(1- \varepsilon_j) Q_o = cd + \varepsilon_j a_jb_j - L_i  \quad \text{and} \quad (1- \varepsilon_{j'}) Q_o = ef + \varepsilon_{j'}a_{j'}b_{j'} - L_i \;.
\end{equation*} 
As $\rank_s(Q_o) \geq \rkq$ 
it must hold that $\varepsilon_j = \varepsilon_{j'} = 1$ and hence
\begin{equation*}
L_i = cd +   a_jb_j \quad \text{and} \quad 
L_i = ef + a_{j'}b_{j'} \;.
\end{equation*} 
It follows that  $cd -ef =   a_{j'}b_{j'}- a_{j}b_{j}$ and therefore $\MS(cd -ef) \subseteq V$. \autoref{cla:rank-2-in-V} implies that without loss of generality $d = \epsilon_i c + v_i$. We therefore conclude that $$P_i = Q_o + L_i = Q_o + a_jb_j + c( \epsilon_i c + v_i)$$ and the statement holds for $P'_i = a_jb_j$ and $c_i = c$.
This completes the proof of the \autoref{cla:V-map-Q-dom}.
\end{proof}

Consider the representation guaranteed in \autoref{cla:V-map-Q-dom} and let 
\[\begin{split}
\calS = \lbrace c_i \mid \text{there is } P_i\in\cQ \text{ such that either } P_i = c_i^2 \text{ or, for some } P'_i \text{ defined over } V, \\  P_i= Q_o +P'_i + c_i(\epsilon_i c_i + v_i)  \rbrace \;.
\end{split}\] 
Clearly, in order to bound the dimension of $\tilde{\cQ}$ it is enough to bound the dimension of $\calS$. We do so, by proving that $\calS$ satisfies the conditions of \hyperref[thm:SG-linforms]{Sylvester-Gallai theorem} modulo $V$, and thus have dimension at most $3+\dim(V)=7$.

\begin{claim} \label{cla:sg-mod-v}
Let $c_i, c_j \in \calS$ be such that $c_i \notin V$ and $c_j \notin \spn{c_i,V}$. Then, there is $c_k \in \calS$ such that $c_k \in \spn{c_i,c_j,V}$ and $c_k \notin \spn{c_i,V}\cup \spn{c_j,V}$.
\end{claim}

Before proving the claim we prove the following simple lemma.

\begin{lemma}\label{lem:ef}
Let $P_V$ be a polynomial defined over $V$ and let $c_i,c_j$ as in \autoref{cla:sg-mod-v}. If there are linear forms $e,f$ such that $$c_j(\epsilon_j c_j + v_j) + c_i(\epsilon_i c_i + v_i) + ef = P_V$$ then, without loss of generality,  $e \in \spn{c_i,c_j,V}$ and $e \notin \spn{c_i,V}\cup \spn{c_j,V}$.
\end{lemma}

\begin{proof}

First note that $e\not\in V$ as otherwise we would have that $c_i\equiv_V c_j$ in contradiction.

By our assumption, $ef =P_V$ modulo $c_i,c_j$.  We can therefore assume without loss of generality that  $e \in \spn{c_i,c_j,V}$. 
Assume towards a contradiction and without loss of generality that  $e = \lambda c_i + v_e$, where $\lambda \neq 0$ and $v_e\in V$. Consider the equation $c_j(\epsilon_j c_j + v_j) + c_i(\epsilon_i c_i + v_i) + ef = P_V$ modulo $c_i$. We have that $c_j(\epsilon_j c_j + v_j) + v_e f \equiv_{c_i} P_V$ which implies that $\epsilon_j=0$. Consequently, we also have that  $f = \mu c_j + \eta c_i + v_f$, for some $\mu \neq 0$ and $v_f \in V$. We now observe that the product $c_ic_j$ has a non zero coefficient $\lambda\mu$ in $ef$ and a zero coefficient in $P_V - c_j(\epsilon_j c_j + v_j) + c_i(\epsilon_i c_i + v_i) $, in contradiction.
\end{proof}

\begin{proof}[Proof of \autoref{cla:sg-mod-v}] 

Following the notation of \autoref{cla:V-map-Q-dom}, we either have $Q_i = Q_o + Q'_i +  c_i(\epsilon_i c_i + v_i )$ or   $Q_i = c_i^2$. Very similarly to \autoref{cla:stisfy-sg-lines-hr}, we  consider which case of \autoref{thm:structure}  $Q_i$ and $Q_j$  satisfy, and what structure they have.

Assume  $Q_i = Q_o + Q'_i +  c_i(\epsilon_i c_i + v_i )$ and  $Q_j = Q_o + Q'_j +  c_j(\epsilon_jc_j + v_j )$. As argued before, since the rank of $Q_o$ is large  they can not satisfy
 \autoref{thm:structure}\ref{case:2}. We consider the remaining cases:
\begin{itemize}
\item $Q_i,Q_j$ satisfy \autoref{thm:structure}\ref{case:span}: there is $Q_k \in \cQ$ such that $Q_k \in \spn{Q_i,Q_j}$. 

By assumption, for some scalars $\alpha,\beta$ we have that 
\begin{equation}\label{eq:Qk}
Q_k  = \alpha(Q_o + Q'_i +  c_i(\epsilon_i c_i + v_i )) + \beta(Q_o + Q'_j +  c_j(\epsilon_j c_j + v_j)) \;.  
\end{equation}
If $Q_k$ depends only on $V$ then we would get a contradiction to the choice of $c_i,c_j$. Indeed, in this case we have that
$$(\alpha+\beta)Q_o = Q_k - \alpha( Q'_i +  c_i(\epsilon_i c_i + v_i )) - \beta( Q'_j +  c_j(\epsilon_j c_j + v_j))\;.$$
Rank arguments imply that $\alpha+\beta=0$ and therefore 
$$\alpha c_i(\epsilon_i c_i + v_i ) + \beta c_j(\epsilon_j c_j + v_j) = Q_k - \alpha Q'_i - \beta Q'_j \;,$$
which implies that $c_i$ and $c_j$ are linearly dependent modulo $V$ in contradiction. 

If $Q_k = c_k^2$ then by \autoref{lem:ef} it holds that $c_k$ satisfies the claim condition.

We therefore assume that $Q_k$ is not a function of $V$ alone and denote $Q_k = Q_o + Q'_k +  c_k(\epsilon_k c_k + v_k)$.
\autoref{eq:Qk} implies that $$(1 - \alpha-\beta)Q_o = \alpha Q'_i + \beta Q'_j - Q'_k + \alpha c_i(\epsilon_i c_i + v_i ) + \beta c_j(\epsilon_j c_j + v_j) - c_k(\epsilon_k c_k + v_k)\;.$$ As $\alpha Q'_i + \beta Q'_j - Q'_k$ is a polynomial defined over $V$, its rank is smaller than $4$ and thus, combined with the fact that $\rank_s(Q_o) \geq \rkq$, we get that $ (1 - \alpha-\beta) = 0$ and 
\[Q'_k -\alpha Q'_i - \beta Q'_j = \alpha c_i(\epsilon_i c_i + v_i ) + \beta c_j(\epsilon_j c_j + v_j) - c_k(\epsilon_k c_k + v_k) \;.\]
We now conclude from \autoref{lem:ef} that $c_k$ satisfies the claim. 

\item $Q_i,Q_j$ satisfy \autoref{thm:structure}\ref{case:rk1}: 
There are linear forms $e,f$ such that for non zero scalars $\alpha,\beta$, $\alpha Q_i + \beta Q_j = ef$. In particular, 
\[(\alpha + \beta ) Q_o = ef - \alpha Q'_i - \beta Q'_j - \alpha c_i(\epsilon_i c_i + v_i ) - \beta c_j(\epsilon_j c_j + v_j).\]
From rank argument we get that $\alpha+ \beta = 0$ and from \autoref{lem:ef} we conclude that, without loss of generality, $e = \mu c_i + \eta c_j +v_e$ where $\mu, \eta \neq 0$. We also assume without loss of generality that $Q_i=Q_j+ef$.

By our assumption that $\rank_s(Q_o)\geq \rkq$ it follows that $Q_j$ is irreducible even after setting $e=0$. It follows that if a product of irreducible quadratics satisfy $$\prod_k A_k \in \sqrt{\ideal{Q_i,Q_j}} = \sqrt{\ideal{ef,Q_j}}$$ then, after setting $e=0$, some $A_k$ is divisible by ${Q_j}|_{e=0}$. Thus, there is a multiplicand that is equal to $ \gamma Q_j + ed$ for some linear form $d$ and scalar $\gamma$. In particular, there must be a polynomial $Q_k \in \tilde\cQ\setminus\{Q_1,Q_2\}$, such that  $Q_k = \gamma Q_j + ed$. If $\gamma = 0$ then it must hold that $Q_k = a_k^2 = ed$ and thus $a_k \sim e$, and the statment holds. If $\gamma = 1$ then we can assume without loss of generality that $Q_k =  Q_j + ed$. Thus,
$$ Q + Q'_k + c_k(\epsilon_k c_k + v_k) = Q_k =  Q_j + ed = Q_o+ Q'_j +  c_j(\epsilon_j c_j + v_j) + (\mu c_i + \eta c_j + v_e)d \;.$$ 
Setting $c_j=0$ we get that 
\begin{equation}\label{eq:d}
Q'_k + c_k(\epsilon_k c_k +v_k) \equiv_{c_j} Q'_j + (\mu c_i + v_e) d \;.
\end{equation}
Note that it cannot be the case that $d\equiv_{c_j}0$. Indeed, if $d=0$ then we get that $Q_j$ and $Q_k$ are linearly dependent in contradiction. If $d \sim c_j$ then  \eqref{eq:d} implies that $c_k\in\spn{c_j,V}$. From the equality  $Q_k = Q_j + ed$ and the fact that $e$ depends non trivially on $c_i$, it now follows that $c_i \in \spn{c_j,V}$ in contradiction to the choice of $c_i$ and $c_j$. As $d\not\equiv_{c_j}0$, we deduce from \eqref{eq:d} that, modulo $c_j$, $c_k \in \spn{c_i, V}$. We next show that if $c_k$ depends only on $c_i$ and $V$ then we reach a contradiction and this will conclude the proof. So assume towards a contradiction that $c_k = \lambda c_i + v'_k$, for a scalar $\lambda$ and $v'_k\in V$.  Since $$Q_j + ed=Q_k = Q_o+ Q'_k + c_k(\epsilon_k c_k + v_k ) = Q_o+Q'_k+ ( \lambda c_i + v'_k)\left(\epsilon_k  ( \lambda c_i + v'_k) + v_k\right) $$ and $$Q_j + ef = Q_i = Q_o +Q'_i + c_i(\epsilon_i c_i + v_i) $$ we get by subtracting $Q_i$ from $Q_k$ that 
\[
  e(d-f)   = Q_k-Q_i = Q'_k - Q'_i + ( \lambda c_i + v'_k)\left(\epsilon_k  ( \lambda c_i + v'_k) + v_k\right) -c_i(\epsilon_i c_i + v_i) 
\]
and clearly neither side of the equation is zero since $Q_i\neq Q_k$. 
This implies that $e \in \spn{c_i,V}$. This however contradicts the fact that  $e = \mu c_i + \eta c_j +v_e$ where $\mu, \eta \neq 0$. 
\end{itemize}

Now let us consider the case where without loss of generality, $Q_i = Q_o +Q'_i + c_i(\epsilon_i c_i + v_i)$ and $Q_j = c_j^2$. In this case the polynomials satisfy   \autoref{thm:structure}\ref{case:rk1} as $0 \cdot Q_i + Q_j = c_j^2$.
Similarly to the previous argument, it holds that there is $Q_k$ such that $Q_k = \gamma Q_i + c_j e$. If $\gamma =0$ it holds that $Q_k$ is reducible, and therefore a square of a linear form, in contradiction to pairwise linear independence. Thus $\gamma \neq 0$. If $Q_k$ is defined only on the linear functions in $V$ then it is of rank smaller then $\dim(V) \leq 4$, which will result in a contradiction to the rank assumption on $Q_o$. Thus $Q_k = Q_o + Q'_k + c_k(\epsilon_k c_k + v_k)$ and $\gamma = 1$. Therefore, we have 
\[ Q_o + Q'_k + c_k(\epsilon_k c_k + v_k) = Q_k = Q_i + c_j e\; = Q_o + Q'_i + c_i(\epsilon_i c_i + v_i)+ c_j e.\]
Hence,
\[Q'_k - Q'_i -c_i(\epsilon_i c_i + v_i) - c_je= -c_k(\epsilon_k c_k + v_k).\] 

Looking at this equation modulo $c_j$ implies that $c_k \in \spn{V, c_i, c_j}$. and $c_k \notin \spn{V, c_j}$, or we will get a contradiction to the fact that $c_i \notin\spn{c_j ,V}$. Similarly it holds that $c_k \notin \spn{V, c_i}$, as we wanted to show.

The last structure we have to consider is the case where $Q_i = c_i^2, Q_j = c_j^2$. In this case, the ideal $\sqrt{\ideal{c_i^2,c_j^2}} = \ideal{c_i,c_j}$ is prime and therefore there is $Q_k \in \ideal{c_i,c_j}$ this means that $\rank_s(Q_k) \leq 2$. If $\rank_s(Q_k) = 1$ then $Q_k = c_k^2$ and the statement holds. $\rank_s(Q_k) = 2$ then $Q_k$ is defined on the linear function of $V$,  which implies $c_i,c_j \in V$ in contradiction to our assumptions.
\end{proof}
We are now ready to prove \autoref{cla:strong-Q-dom-hr-main}.
\begin{proof}[Proof of \autoref{cla:strong-Q-dom-hr-main}]
\autoref{cla:sg-mod-v} implies that if we project the linear forms in $\calS$ to $V^\perp$ then, after removing linearly dependent forms, they satisfy the conditions of the \hyperref[thm:SG-linforms]{Sylvester-Gallai theorem}. As $\dim(V)\leq 4$ we obtain that $\dim(\spn{\calS\cup V}) \leq 7$. 
By \autoref{cla:V-map-Q-dom} every polynomial $P\in\cQ$ is a linear combination of $Q_o$ and a polynomial defined over $\spn{\calS\cup V}$ which, by the argument above, implies that $\dim(\spn{\cQ})\leq 8$. 
\end{proof}

This completes the proof of  \autoref{thm:Q-dom-gen} when  $Q_o$ has high rank.
We next handle the case where $Q_o$ is of low rank. 

\subsection{$Q_o$ is of Low Rank}
%


In this section we prove the following claim.

\begin{claim}\label{cla:strong-Q-dom-lr-main}
Let $\tilde{\cQ}$ be a \hyperref[def:Q-dom]{$(Q_o,m_1,m_2)$-set} such that $2\leq \rank_s(Q_o)< \rkq$.
Then, $\dim(\spn{\tilde{\cQ}}) = O(1)$.
\end{claim}

Before we start with the proof of the main claim, let us prove a similar claim but for a more specific structure of polynomials. We will later see that, essentially, this structure holds when $2\leq \rank_s(Q_o)< \rkq$.

\begin{claim}\label{cla:V-li}
	Let $\tilde{\cQ}$ be a set of quadratics polynomials that satisfy the conditions in the statement of \autoref{thm:main-sg-intro}. Assume farther that there is a linear space of linear forms, $V$ such that $\dim(V)=\Delta$ and for each polynomial $Q_i \in \tilde{\cQ}$ one of the following holds: either $Q_i \in \ideal{V}$ or there is a linear form $a_i$ such that $\MS(Q_i) \subseteq \spn{V,a_i}$. Then $\dim(\tilde{\cQ}) \leq 8\Delta^2 $.
\end{claim}

\begin{proof}
	Note that by the conditions  in the statement of \autoref{thm:main-sg-intro}, no two polynomials in $\tilde{\cQ}$ share a common factor.
	
	Let $\vaa\in\C^\Delta$ be such that if two polynomials in $T_{\vaa, V}(\tilde{\cQ})$ (recall \autoref{def:z-mapping}) share a common factor then it is a polynomial in $z$. Note that by \autoref{cla:res-z-ampping} such $\vaa$ exists. Thus, each $P\in \tilde{\cQ}$, satisfies that either $T_{\vaa, V}(P)=\alpha_P z^2$ or $\MS(T_{\vaa, V}(P))\subseteq \spn{z,a_P}$ for some linear form $a_P$ independent of $z$.
	It follows that every polynomial in $T_{\vaa, V}(\tilde{\cQ})$ is reducible. We next show that $\calS = \{a_P \mid P\in \tilde\cQ\}$ satisfies the conditions of \hyperref[thm:SG-linforms]{Sylvester-Gallai theorem} modulo $z$. 
	
	Let $a_1, a_2 \in \calS$ such that $a_2 \notin \spn{z,a_1}$. Consider $Q_1$ such that $\MS(T_{\vaa, V}(Q_1)) \subseteq \spn{z,a_1}$ yet $\MS(T_{\vaa, V}(Q_1)) \not\subseteq \spn{z}$. Similarly, let $Q_2$ be such that $\MS(T_{\vaa, V}(Q_2)) \subseteq \spn{z,a_2}$ and $\MS(T_{\vaa, V}(Q_2)) \not\subseteq \spn{z}$.
	Then there is a factor of $T_{\vaa,V}(Q_1)$ of the form $\gamma_1 z + \delta_1 a_1$ where $\delta_1\neq 0$. Similarly  there is a factor of $T_{\vaa,V}(Q_2)$ of the form $\gamma_2 z + \delta_2 a_2$ where $\delta_2\neq 0$.
	
	This implies that $\sqrt{\ideal{T_{\vaa,V}(Q_1), T_{\vaa,V}(Q_2)}} \subseteq \ideal{\gamma_1 z + \delta_1 a_1, \gamma_2 z + \delta_2 a_2}$. Indeed, it is clear that for $i\in\{1,2\}$, ${T_{\vaa,V}(Q_i)}\in\ideal{\gamma_i z + \delta_i a_i}$. Hence, $\sqrt{\ideal{T_{\vaa,V}(Q_1), T_{\vaa,V}(Q_2)}} \subseteq \sqrt{\ideal{\gamma_1 z + \delta_1 a_1, \gamma_2 z + \delta_2 a_2}}=\ideal{\gamma_1 z + \delta_1 a_1, \gamma_2 z + \delta_2 a_2}$, where the equality holds since $\ideal{\gamma_1 z + \delta_1 a_1, \gamma_2 z + \delta_2 a_2}$ is a prime ideal.
	
	 We know that, there are $Q_3,Q_4,Q_5,Q_6\in \cQ$ such that 
	\[
	Q_3\cdot Q_4\cdot Q_5\cdot Q_6 \in \sqrt{\ideal{Q_1,Q_2}}.
	\] 
	As $T_{\vaa,V}$ is a ring homomorphism it follows that,
	\[
	T_{\vaa,V}(Q_3)\cdot T_{\vaa,V}(Q_4)\cdot T_{\vaa,V}(Q_5)\cdot T_{\vaa,V}(Q_6)\in \sqrt{\ideal{T_{\vaa,V}(Q_1),T_{\vaa,V}(Q_2)}} \subseteq \ideal{\gamma_1 z + \delta_1 a_1, \gamma_2 z + \delta_2 a_2}.
	\] 
	
	Since $\ideal{\gamma_1 z + \delta_1 a_1, \gamma_2 z + \delta_2 a_2}$ is prime it follows that, without loss of generality, $T_{\vaa,V}(Q_3)\in\ideal{\gamma_1 z + \delta_1 a_1, \gamma_2 z + \delta_2 a_2}$. It cannot be the case that $T_{\vaa,V}(Q_3)\in\ideal{\gamma_i z + \delta_i a_i}$ for any $i \in \{1,2\}$, because otherwise this will imply that $T_{\vaa,V}(Q_3)$ and $T_{\vaa,V}(Q_i)$ share a common factor that is not a polynomial in $z$, in contradiction to our choice of $T_{\vaa, V}$. This means that there is a factor of $T_{\vaa,V}(Q_3)$ that is in $\spn{a_1,a_2,z} \setminus \left(\spn{a_1,z}\cup \spn{a_2,z}\right)$.
	Consequently, $a_3 \in \spn{a_1,a_2,z} \setminus \left(\spn{a_1,z}\cup \spn{a_2,z}\right)$ as we wanted to prove.
	This shows that $\calS$ satisfies the conditions of \hyperref[thm:SG-linforms]{Sylvester-Gallai theorem}, and therefore $\dim(\calS)\leq 3$. Repeating the analysis above for linearly independent $\MVar{\vaa}{\Delta}$, we can use \autoref{cla:z-map-dimension} and obtain that $\dim(\MS(\tilde\cQ)) \leq (3+1)\Delta$, and thus $\dim(\tilde{\cQ}) \leq {4\Delta \choose 2}+\Delta \leq 8\Delta^2$. 
\end{proof}

Back to the proof of \autoref{cla:strong-Q-dom-lr-main}. As before we first prove the claim for the case $m_2=0$ and then we prove the general case.

\subsubsection{The case $m_2=0$}

Similarly to the high rank case, in this subsection we prove the following claim.
\begin{claim}\label{cla:strong-Q-dom}
Let $\tilde{\cQ} = \cQ\cup \calL$ be a\hyperref[def:Q-dom]{$(Q_o,m_1,0)$-set} such that $2\leq \rank_s(Q_o)< \rkq$,
then  $\dim(\spn{\MVar{a}{m_1},\MVar{b}{m_1},\MVar{\ell}{r}}) = O(1)$.
\end{claim}

The proof is similar in structure to the proof of \autoref{cla:strong-Q-dom-hr}.
As before, we consider a polynomial $\ell_i^2\in \calL$ as $0\cdot Q_o+ \ell_i\ell_i$.
We start by proving an analog of \autoref{cla:span-rank1-hr}. The claims are similar but the proofs are slightly different as we cannot rely on $Q_o$ having high rank. 

\begin{claim}\label{cla:span-rank1}
Let $\tilde\cQ$ satisfy  the assumptions of \autoref{cla:strong-Q-dom}. Let $i\in[m_1]$ be such that $dim(a_i,b_i) = 2$ and $\spn{a_i,b_i} \cap \MS(Q_o) = \{\vec{0}\}$. Then, for every $j\in [m_1]$ the following holds:
\begin{enumerate}
\item $Q_i$ and $Q_j$ do not satisfy \autoref{thm:structure}\ref{case:2}.

\item If $Q_i$ and $Q_j$ satisfy \autoref{thm:structure}\ref{case:span} then there exists $\alpha, \beta\in\C\setminus\{0\}$ such that for some $k\in[m_1]\setminus \{i,j\}$
\begin{equation} \label{eq:1-alpha}
 \alpha a_ib_i + \beta a_jb_j = a_k b_k \;. 
\end{equation}
\item If $Q_j$ is irreducible and $Q_i$ and $Q_j$ satisfy \autoref{thm:structure}\ref{case:rk1} then there exist two linear forms, $c$ and $d$ such that   
\begin{equation}
 a_ib_i -  a_jb_j = cd \;.
\end{equation}
\end{enumerate}
\end{claim}

\begin{proof}
Assume $Q_i$ and $Q_j$ satisfy \autoref{thm:structure}\ref{case:span}, i.e., there are $\alpha, \beta \in \C$ and $k \in [m_1]\setminus \{i,j\}$ such that 
$$\alpha (Q_o+a_ib_i) + \beta ( Q_o+a_jb_j) = \alpha Q_i + \beta Q_j = Q_k =\alpha_kQ+a_kb_k \;$$ 
This implies that $\alpha a_ib_i + \beta a_jb_j -a_kb_k = (\alpha_k-(\alpha+\beta))Q_o$.
We next show that it must be the case that $\alpha_k-(\alpha+\beta) = 0$. 

Indeed, if $\alpha_k-(\alpha+\beta)\neq 0$ we get that $\beta a_jb_j -a_kb_k = (\alpha_k-(\alpha+\beta))Q_o- \alpha a_ib_i$. However,  as we assumed 
$\spn{a_i,b_i} \cap \MS(Q_o) = \{\vec{0}\}$, we get by  \autoref{cla:ind-rank} that 
$$\rank_s(\alpha_k-(\alpha+\beta))Q_o- \alpha a_ib_i)= \rank_s(Q_o)+1 > 2 \geq \rank_s(\beta a_jb_j -a_kb_k)$$ in contradiction.
We thus have that $\alpha_k-(\alpha+\beta) = 0$ and hence
\begin{equation} \label{eq:span-ak}
 \alpha a_ib_i + \beta a_jb_j = a_kb_k
\end{equation}
and \autoref{eq:1-alpha} is satisfied. Observe that since our polynomials are pairwise independent $\alpha, \beta\neq 0$.

A similar argument to the one showing $\alpha_k-(\alpha+\beta) = 0$ also implies that  $Q_i$ and $Q_j$ do not satisfy \autoref{thm:structure}\ref{case:2}. If this was not the case then we would have that $\rank_s(Q_o+a_ib_i)= 2$ which would again contradict  \autoref{cla:ind-rank}.

If $Q_j$ is irreducible, the only case left is when $Q_o+a_ib_i, Q_o+a_jb_j$ satisfy  \autoref{thm:structure}\ref{case:rk1}. In this case there are $\alpha, \beta \in \C$ and two linear forms $c$ and $d$ such that $ \alpha (Q_o+a_ib_i) + \beta ( Q_o+a_jb_j) = cd$, and again, by the same argument we get that  $\beta = -\alpha$ and so (after rescaling $c$)
\begin{equation*}
  a_ib_i -a_jb_j = cd \;. 
\end{equation*}
This completes the proof of \autoref{cla:span-rank1}.
\end{proof}

For each $i\in [m_1]$ let $V_i \eqdef \spn{a_i,b_i}$. 
The next claim is analogous to \autoref{cla:V_i intersection-hr}.
 
 \begin{claim}\label{cla:V_i intersection}
Let $\tilde\cQ$ satisfy the assumption in \autoref{cla:strong-Q-dom}. If for some $i\in [m_1]$ it holds that $\dim(V_i)=2$ and $\MS(Q_o)\cap V_i = \{\vec{0}\}$ then for every $j \in [m_1]$ it is the case that $\dim(V_j\cap V_i) \geq 1$. In particular, if $dim(V_j)=1$ then $V_j\varsubsetneq V_i$.
\end{claim}
\begin{proof}
The proof of this claim follows immediately from \autoref{cla:span-rank1} and \autoref{cla:intersection}.
\end{proof}

the next claim is an analogous to \autoref{cla:z-exist}.

\begin{claim}\label{cla:V-exist}
Under the assumptions of \autoref{cla:strong-Q-dom} there exists a subspace $V$ of linear forms such that $\dim(V)\leq 2\cdot\rkq+ 3$ and  for every $i \in [m_1]$ there exists $v_i\in V$ and a constant $\epsilon_i\in\C$ such that $b_i = \epsilon_i a_i +v_i$ (or $a_i = \epsilon_i b_i +v_i$).
\end{claim}

\begin{proof}

Let $\cI=\{i\in[m_1] \mid \dim(V_i)=2 \text{ and } \MS(Q_o)\cap V_i = \{\vec{0}\} \}$. 
If $\dim(\bigcup_{i\in \cI} V_i) \leq 3$ then we set $V = \spn{\MS(Q_o)  \cup (\bigcup_{i\in \cI} V_i)}$. Clearly $\dim(V)\leq 2\cdot \rank_s(Q)+3\leq 2\cdot \rkq+3$.  \autoref{cla:V_i intersection}  implies that $V$ has the required properties. 

If $\dim(\bigcup_{i\in \cI} V_i) > 3$ then from \autoref{cla:V_i intersection} and \autoref{cla:linear-spaces-intersaction }
it follows that $\dim(\bigcap_{i\in \cI} V_i) =1$.
Let $w$ be such that $\spn{w}=\bigcap_{i\in \cI} V_i$ and set $V = \spn{\MS(Q_o),w}$. 
In this case too it is easy to see  that $V$ has the required properties.
\end{proof}


From now on we assume, without loss of generality that for every $i\in [m_1]$, $b_i = \epsilon_i a_i +v_i$. This structure also holds for the polynomials in $\calL$.

\begin{proof}[Proof of \autoref{cla:strong-Q-dom}]
	\autoref{cla:V-exist} implies that there is a linear space of linear forms, $V$, with $\dim(V) \leq 2\cdot\rkq+ 3$, with the property that for every $Q_i \in \tilde{\cQ}$  there is a linear form $a_i$ such that $\MS(Q_i) \subseteq \spn{V,a_i}$. 
	Thus $\tilde{\cQ}$ satisfies the conditions of \autoref{cla:V-li}, and $\dim(\tilde{\cQ}) = O(1)$, as we wanted to show.
\end{proof}

We next consider the case $m_2\neq 0$.

\subsubsection{The case $m_2\neq 0$}

In this subsection we prove  \autoref{cla:strong-Q-dom-lr-main}, we can assume without loss of generality that $m_2 \neq 0$, as the case that $m_2=0$ was proved in the previous subsection.
To handle this case we prove the existence of a subspace $V$ of linear forms, of dimension $O(1)$, such that every polynomial in $\tilde{\cQ}$ is in $\ideal{V}$, and then, like we did before, we bound the dimension of $\tilde{\cQ}$. The first step is proving an analog of \autoref{cla:no-rank-3}.

\begin{claim}\label{cla:Q-and-rank-2}
Let $\tilde{\cQ}$ be a \hyperref[def:Q-dom]{$(Q_o,m_1,m_2)$-set} such that $\rank_s(Q_o)< \rkq$. Then for every $i \in [m_2]$ there exists $ \gamma_i\in \C$  such that $\rank_s(P_i-\gamma_i Q_o) = 2$.
\end{claim}

\begin{proof}
Consider $i\in [m_2]$. If $P_i$ satisfies \autoref{thm:structure}\ref{case:2} with any $Q_j\in \cQ_1$, then the claim holds with $\gamma_i=0$. 
If $P_i$ satisfies \autoref{thm:structure}\ref{case:rk1} with any $Q_j\in \cQ$ then there exist linear forms  $c$ and $d$ and non zero $\alpha, \beta \in \C$, such that $\alpha P_i + \beta Q_j = cd$. Therefore, $
P_i = \frac{1}{\alpha}(cd - \beta (Q+a_j b_j))$
and the statement holds with $\gamma_i = -\frac{\beta}{\alpha}$. Observe that the rank of $cd - \beta a_j b_j$ cannot be $1$ by \autoref{def:Q-dom}.

Thus, the only case left to consider is when $P_i$ satisfies  \autoref{thm:structure}\ref{case:span} with all the $Q_j$'s in $\cQ_1$. 
We next show that in this case there must exist $j\neq j' \in [m_1]$ such that $Q_{j'}\in \spn{Q_j,P_i}$. Indeed,  since $m_1 >  5m_2+2$ there must be $j,j'\in[m_1]$ and $i'\in[m_2]$ such that
$P_{i'} \in \spn{Q_{j'}, P_i}$ and $ P_{i'} \in \spn{Q_{j}, P_i}$. As we saw before this implies that $P_i \in \spn{Q_j,Q_{j'}}$, which is what we wanted to show.

Let  $j \neq j' \in [m_1]$ be as above and let $\alpha, \beta \in \C$ be such that $P_i =\alpha Q_j + \beta Q_{j'}$.  It follows that
\[
P_i = (\alpha + \beta)Q_o+ \alpha a_j b_j + \beta a_{j'} b_{j'}\;.
\] 
Let $\gamma_i = \alpha + \beta$. Property~\ref{item:Qdom>2} in \autoref{def:Q-dom} implies that $\rank_s(\alpha a_j b_j + \beta a_{j'} b_{j'})=2$ and the claim follows.
\end{proof}

As before, whenever $\gamma_i\neq 0$ let us replace $P_i$ with $\frac{1}{\gamma_i}P_i$. Thus, from now on we shall assume $\gamma_i \in \{0,1\}$.
We next prove an analog of \autoref{cla:small-space-span}.

\begin{claim}\label{cla:m-2k-in-ideal-v}
Let $\tilde{\cQ}$ be a \hyperref[def:Q-dom]{$(Q_o,m_1,m_2)$-set} such that $\rank_s(Q_o)< \rkq$. 
Then there is a subspace $V$ of linear  forms such that $\dim(V) \leq 2\cdot \rkq+4$, $\MS(Q_o)\subseteq V$ and for at least $m_1-2m_2 $ of the indices $j\in[m_1]$ it holds that $a_j,b_j \in V$.
\end{claim}
\begin{proof}
Let $P=P_1$.
\autoref{cla:Q-and-rank-2} implies that $P = \gamma Q_o+L$, for some $L$ of rank $2$. Set $V = \spn{\MS(Q_o)\cup\MS(L)}$. Clearly $\dim(V)\leq 2 \cdot \rkq+4$.

Let $j\in [m_1]$. 
If $P$ and $Q_j$ satisfy \autoref{thm:structure}\ref{case:2}, then there are two linear forms $c$ and $d$ such that  $Q_j,P\in \sqrt{\ideal{c,d}}$, this implies that $\spn{c,d} \subset\MS(P) \subseteq V$. 
If $Q_o=Q_j-a_j b_j$ is not zero modulo ${c,d}$, then we obtain that $Q_o\equiv_{c,d} -a_j b_j$. Thus, there are linear forms $v_1,v_2\in\MS(Q_o)$ such that $a_j \equiv_{c,d} v_1$ and $b_j \equiv_{c,d} v_2$. In particular, as $\MS(Q_o)\cup\{c,d\}\subset V$ it follows that $a_j,b_j\in V$.
If $Q_o$ is zero modulo $c$ and $d$, then $Q_{j},Q_o$ satisfy \autoref{thm:structure}\ref{case:2} and from property~\ref{item:Qdom-case3} of \autoref{def:Q-dom} we know that there are at most $m_2$ such $Q_j$'s. Furthermore, as  $c,d\in\MS(Q_o) \subset V$ we obtain that $Q_j \in \ideal{V}$. Denote by $\cK$ the set of all $Q_j$ that satisfy \autoref{thm:structure}\ref{case:2} with $Q_o$. As we mentioned, $|\cK|\leq m_2$.

If $P$ and $Q_j$ satisfy  \autoref{thm:structure}\ref{case:rk1} then there are two linear forms $c$ and $d$, and non zero $\alpha, \beta \in \C$, such that $\alpha P + \beta Q_j = cd$. Hence,
\begin{equation*}
\beta Q_o+\alpha P=  - \beta a_jb_j +cd\;.
\end{equation*}
As $\beta Q_o+\alpha P$ is a non trivial linear combination of $Q_o$ and $P$, we get from property~\ref{item:Qdom>2} of \autoref{def:Q-dom} that $2 \leq \rank_s((\alpha\gamma+\beta)Q_o+\alpha L)$. It follows that $$\rank_s(- \beta a_jb_j +cd)=\rank_s((\alpha\gamma+\beta)Q_o+\alpha L)= 2$$ and therefore by \autoref{cor:containMS}, $$\{a_j,b_j,c,d\}\subset \MS(-\beta a_jb_j +cd)= \MS((\alpha\gamma+\beta)Q_o+\alpha L)\subseteq V\;,$$ and again $a_j,b_j\in V$.

The last case to consider is when $P$ and $Q_j$ satisfy  \autoref{thm:structure}\ref{case:span}. If they span a polynomial $Q_{j'} \in \cQ_1\cup \calL$,   then $P = \alpha Q_j + \beta Q_{j'}$ and as in the previous case we get that $ a_{j}, b_{j} \in V$.

Let $\cJ$ be the set of all indices $j\in [m_1]$ such that $P$ and $Q_j$ span a polynomial in $\cQ_2$ but no polynomial in $\cQ_1\cup \calL$. So far we proved that  for every $j\in [m_1]\setminus (\cJ\cup \cK)$ we have that $a_j,b_j\in V$. We next show that $|\cJ|\leq m_2$ which concludes the proof. 

Indeed, if this was not the case then by the pigeonhole principle there would exist a polynomial $P_i\in\cQ_2$ and two polynomials $Q_j,Q_{j'}\in\cQ_1 $ such that $P_i\in\spn{Q_j,P}$ and $P_i\in\spn{Q_{j'},P}$. By pairwise independence this implies that $Q_{j'}$ is in the linear span of $P$ and $Q_j$ which contradicts the definition of $\cJ$.
\end{proof}

Our next claim gives more information about the way the polynomials in $\tilde\cQ$ relate to the subspace $V$ found in \autoref{cla:m-2k-in-ideal-v}.

\begin{claim}\label{cla:mod-V}
Let $\tilde{\cQ}$ and $V$ be as in \autoref{cla:m-2k-in-ideal-v}. Then,  every polynomial $P$ in $\tilde{\cQ}$ satisfies (at least) one of the following cases:
\begin{enumerate}
\item $\MS(P)\subseteq V$ or
\item $P\in \ideal{V}$ or
\item  $P=P'+c(c + v)$ where $P'$ is a quadratic polynomial such that $\MS(P')\subseteq V$, $v\in V$ and $c$ is a linear form.
\end{enumerate}
\end{claim}

\begin{proof}
Let $\cI = \lbrace j\in [m_1] \mid  a_j,b_j \in V\rbrace$.  \autoref{cla:m-2k-in-ideal-v} implies that $|\cI|\geq m_1-2m_2$. Furthermore, by the construction of $V$ we know that $\MS(Q_o)\subseteq V$. Observe that this implies that for every $j\in \cI$, $\MS(Q_j)\subseteq V$.

Note that every polynomial in $\calL$ satisfies the third item of the claim.
Let $P$ be any polynomial in $\cQ_2 \cup \{Q_j \mid j\in [m_1]\setminus \cI\}$. We study which case of \autoref{thm:structure} $P$ satisfies with polynomials whose indices belong to $\cI$.

If $P_i$ satisfies \autoref{thm:structure}\ref{case:2} with any polynomial $Q_j$, for $j\in \cI$, then, as $\MS(Q_j)\subseteq V$, it follows that  $P\in \ideal{V}$.

If $P$ is spanned by two polynomials $Q_j,Q_{j'}$ such that $j,j'\in \cI$, then clearly $\MS(P)\subseteq V$. Similarly, if $P$ is spanned by a polynomial $Q_j,Q_{j'}$ such that $j \in \cI$ and $Q_{j'}\in \calL$ then $P = \alpha Q_j + \beta a_{j'}^2$, and hence it also satisfies the claim.

Hence, for $P$ to fail to satisfy the claim, it must be the case that every polynomial $Q_j$, for $j\in\cI$, that satisfies \autoref{thm:structure}\ref{case:span} with $P$, does not span with $P$ any polynomial in $\{Q_j \mid j\in  \cI\} \cup \calL$. Thus, it must span with $P$ a polynomial in $\{Q_j \mid j\in [m_1]\setminus  \cI\} \cup \cQ_2$. As before, observe that by pairwise linear independent, if two polynomials from $\cI$ span the same polynomial with $P$, then $P$ is in their span and we are done. Thus, since 
$$\left|\{Q_j \mid j\in [m_1]\setminus  \cI\} \cup \cQ_2\right| \leq \left(m_1 - |\cI| \right) + m_2 \leq 3m_2 < m_1 -2m_2 -2 \leq |\cI| -2 \;,$$
for $P$ to fail to satisfy the claim it must be the case that 
it satisfies \autoref{thm:structure}\ref{case:rk1} with at least $2$ polynomials whose indices are in $\cI$. 

Let $Q_j,Q_{j'}$ be two such polynomials. There are four linear forms, $c,d,e$ and $f$ and scalars $\epsilon_j, \epsilon_{j'}$ such that
\begin{equation*}
P + \varepsilon_j Q_j = cd \quad\quad \text{and}\quad\quad
P + \varepsilon_{j'} Q_{j'} = ef \;.
\end{equation*} 
Therefore
\begin{equation} \label{eq:rank2-in-V}
\varepsilon_j Q_j -\varepsilon_{j'} Q_{j'} = cd -ef \;.
\end{equation} 
In particular, $\MS(cd -ef)\subseteq V$. \autoref{cla:rank-2-in-V} and Equation~\eqref{eq:rank2-in-V} imply that, without loss of generality,  $d = \epsilon c + v$ for some $v \in V$ and $\epsilon \in \C$. 
Thus, $P = cd -\varepsilon_j Q_j = c(\epsilon c+v) -\varepsilon_jQ_j$ and no matter whether $\epsilon=0$ or not. $P$ satisfies the claim. Indeed, if $\epsilon = 0$ then $P \in \ideal{V}$ and we are done. Otherwise, we can normalize $c,v$ to  assume that $\epsilon = 1$ and get that $\MS(P-c^2)\in V$ as claimed. 
\end{proof}

We can now complete the proof of  \autoref{cla:strong-Q-dom-lr-main}.
\begin{proof}[Proof of \autoref{cla:strong-Q-dom-lr-main}]
		\autoref{cla:mod-V} implies that there is a linear space of linear forms, $V$, such that $\dim(V) \leq 2\cdot\rkq+ 4$ and  every polynomial $Q_i \in \tilde{\cQ}$ satisfies the following. Either $Q_i \in \ideal{V}$ or, there is a linear form $a_i$ such that $\MS(Q_i) \subseteq \spn{V,a_i}$. (It might be that $\MS(Q_i) \subseteq V$ or that $\MS(Q_i) \subseteq \spn{a_i}$). Thus $\tilde{\cQ}$ satisfies the conditions of \autoref{cla:V-li}, and $\dim(\tilde{\cQ}) = O(1)$, as we wanted to show.

\end{proof}
\autoref{cla:strong-Q-dom-hr-main} together with \autoref{cla:strong-Q-dom-lr-main} completes the proof of \autoref{thm:Q-dom-gen}.\qed

\section{Conclusions and future research}\label{sec:discussion}

In this work we solved \autoref{prob:sg-alg-prod} in the case where all the polynomials are irreducible and of degree at most  $2$. This result directly relates to the problem of obtaining deterministic algorithms for testing identities of  $\Sigma^{[3]}\Pi^{[d]}\Sigma\Pi^{[2]}$ circuits. As mentioned in \autoref{sec:intro}, in order to obtain PIT algorithms we need a colored version of this result. 
Formally, we need to prove the following conjecture:
\begin{conjecture}\label{thm:main-ek-conc}
Let $\cT_1,\cT_2$ and $\cT_3$ be finite sets of homogeneous quadratic polynomials over $\C$ satisfying the following properties:
\begin{itemize}
\item Each $Q_o\in\cup_i\cT_i$ is either irreducible or a square of a linear form.\footnote{We replace a linear form with its square to keep the sets homogeneous of degree $2$.}
\item No two polynomials are multiples of each other (i.e., every pair is linearly independent).
\item For every two polynomials $Q_1$ and $Q_2$ from distinct sets, whenever $Q_1$ and $Q_2$ vanish then also the product of all the polynomials in the third set vanishes. 
\end{itemize}
Then the linear span of the polynomials in $\cup_i\cT_i$ has dimension $O(1)$.
\end{conjecture}

We believe that tools similar to the  tools developed in this paper should suffice to verify this conjecture.
Another interesting question is a robust version of this problem, which is still open.
\begin{problem}
Let $\delta \in (0,1]$. Can we bound the linear dimension (as a function of $\delta$) of a set of polynomials $Q_1,\ldots ,Q_m\in \CRing{x}{n}$ that satisfy the following property: For every $i\in [m]$ there exist at least $\delta m$ values of $j \in [m]$ such that for each such $j$ there is $\cK_j \subset [m]$, where $i,j \notin \cK_j$ and $\prod_{k\in \cK_j} Q_k \in \sqrt{\ideal{Q_i,Q_j}}$.
\end{problem}

Extending our approach to the case of more than $3$ multiplication gates (or more than $3$ sets as in 
the colored version of the Sylvester-Gallai theorem (\autoref{thm:EK})) seems more difficult. Indeed, an analog of  \autoref{thm:structure} for this case seems harder to prove in the sense that there are many more cases to consider which makes it unlikely that a similar approach will continue to work as the number of gates get larger. Another difficulty is proving an analog of  \autoref{thm:structure} for higher degree polynomials. Thus, we believe that a different proof approach may be needed in order to obtain PIT algorithms for $\Sigma^{[O(1)]}\Pi^{[d]}\Sigma\Pi^{[O(1)]}$ circuits.


In this paper we only considered polynomials over the complex numbers. However, we believe (though we did not check the details) that a similar approach should work over positive characteristic as well. Observe that over positive characteristic we expect the dimension of the set to scale like $O(\log |\cQ|)$, as for such fields a weaker version of Sylvester-Gallai theorem holds.
\begin{theorem}[Corollary $1.3$ in \cite{DBLP:journals/combinatorica/BhattacharyyaDS16}]
Let $V=\lbrace \MVar{\vec{v}}{m}\rbrace\subset \F_p^d$ be a set of $m$ vectors, no two of which are linearly dependent. Suppose that for every $i,j\in[m]$, there exists $k\in[m]$ such that $\vec{v}_i,\vec{v}_j,\vec{v}_k$ are linearly dependent. Then, for every $\epsilon>0$
\[\dim(V)\leq{\rm poly}(p/\epsilon)+(4+\epsilon)\log_{p}m \;.\]
\end{theorem}



\bibliographystyle{customurlbst/alphaurlpp}
\bibliography{main}
\end{document}